\documentclass[a4paper,11pt]{article}

\usepackage{microtype}
\usepackage{mathtools,amsmath}
\usepackage{amssymb,amsthm}
\usepackage[T1]{fontenc}
\usepackage{graphicx}
\usepackage{xspace}
\usepackage{hyperref,color}
\usepackage{pdfpages}
\usepackage{fullpage} 
\usepackage[mathlines]{lineno}
\usepackage{enumerate}
 \newtheorem{theorem}{Theorem} \newtheorem{lemma}[theorem]{Lemma}
 
 \newtheorem{corollary}[theorem]{Corollary}
 
 \newtheorem{definition}[theorem]{Definition}

\newcommand{\xcoord}{\ensuremath{x}}
\newcommand{\ycoord}{\ensuremath{y}}
\newcommand{\myint}[1]{\ensuremath{\textsf{int}(#1)}}

\newcommand{\mybd}[1]{\ensuremath{\partial#1}}
\newcommand{\mycl}[1]{\ensuremath{\textsf{cl}(#1)}}

\newcommand{\norm}[1]{\ensuremath{\left\lVert#1\right\rVert}}
\newcommand{\inner}[2]{\ensuremath{\left\langle #1, #2 \right\rangle}}
\newcommand{\nopt}[1]{\ensuremath{\textsf{opt}(#1)}}

\newcommand{\vv}[1]{\ensuremath{\mathbf{#1}}}

\newcommand{\prob}[1]{\ensuremath{\textsf{wPartition}(#1)}}
\newcommand{\probt}[1]{\ensuremath{\textsf{wPartition}(T, #1)}}
\newcommand{\dwidth}[2]{\ensuremath{\omega_{#1}(#2)}}
\newcommand{\oset}{\ensuremath{\mathcal{O}}}
\newcommand{\bdeset}{\ensuremath{E^\mathrm{bd}}}
\newcommand{\corridor}{\ensuremath{\mathcal{A}}}
\newcommand{\inst}{\ensuremath{\mathtt{INST}}}
\newcommand{\usetp}{\ensuremath{\mathbb{S}^+}}
\newcommand{\paths}{\ensuremath{\Lambda}}
\newcommand{\core}[1]{\ensuremath{\mathrm{core}(#1)}}
\newcommand{\segout}{\ensuremath{\mathcal{Z}^{\mathrm{out}}}}
\newcommand{\segin}{\ensuremath{\mathcal{Z}^{\mathrm{in}}}}
\newcommand{\seginast}{\ensuremath{\mathcal{Z}^{\mathrm{in}\ast}}}
\newcommand{\segunion}{\ensuremath{\mathbf{Z}}}

\newcommand{\trg}[1]{\ensuremath{#1^\mathrm{tr}}}
\newcommand{\nlayers}{\ensuremath{h}}
\newcommand{\conv}[1]{\ensuremath{\mathrm{CH}(#1)}}
\def\denseitems{
    \itemsep1pt plus1pt minus1pt
    \parsep0pt plus0pt
    \parskip0pt\topsep0pt}

\newbox\ProofSym \setbox\ProofSym=\hbox{%
  \unitlength=0.18ex%
  \begin{picture}(10,10) \put(0,0){\framebox(9,9){}}
    \put(0,3){\framebox(6,6){}}
  \end{picture}}

\title{Minimum Partition of Polygons under Width and Cut Constraints\thanks{The work by Chung was supported by a KIAS Individual Grant AP106101 via the Center for Artificial Intelligence and Natural Sciences 
at Korea Institute for Advanced Study. 
The work by Iwama was supported in part by MOST, Taiwan, under
Grants NSTC 110-2223-E-007-001 and NSTC 111-2223-E-007-010.
The work by Liao was supported by MOST Taiwan Grants NSTC 111-2221-E-002-207-MY3, 114-2221-E-002-221-MY3, and 114-2221-E-002-220-MY3.
The work by Ahn was supported by 
the National Research Foundation of Korea (NRF) 
grant funded by the Korea government(MSIT) (RS-2023-00219980) 
and the Institute of Information 
\& communications Technology Planning \& Evaluation(IITP) grant funded by the Korea government(MSIT) (No. RS-2019-II191906, Artificial Intelligence Graduate School Program(POSTECH)).
}}
\author{Jaehoon Chung\thanks{Korea Institute for Advanced Study (KIAS), Seoul, Korea. 
{\tt sk7755@kias.re.kr}} 
  \and Kazuo Iwama\thanks{Department of Industrial Engineering and Engineering Management, National Tsing Hua University, Taiwan. {\tt iwama@ie.nthu.edu.tw}} 
  \and Chung-Shou Liao\thanks{Department of Electrical Engineering, National Taiwan University, Taiwan. {\tt csliao@ntu.edu.tw}} 
  \and Hee-Kap Ahn\thanks{Graduate School of Artificial Intelligence, Department of Computer Science and Engineering, Pohang University of Science and Technology (POSTECH), Korea. {\tt heekap@postech.ac.kr}}}
\begin{document}
\date{}
\maketitle

\begin{abstract}
  We study the problem of partitioning a polygon into the minimum
  number of subpolygons using cuts in predetermined directions such
  that each resulting subpolygon satisfies a given width constraint.
  A polygon satisfies the unit-width constraint for a set of unit
  vectors if the length of the orthogonal projection of the polygon on
  a line parallel to a vector in the set is at most one.  We analyze
  structural properties of the minimum partition numbers, focusing on
  monotonicity under polygon containment.  We show that the minimum
  partition number of a simple polygon is at least that of any
  subpolygon, provided that the subpolygon satisfies a certain
  orientation-wise convexity with respect to the polygon.  As a
  consequence, we prove a partition analogue of Bang's conjecture
  about coverings of convex regions in the plane: for any partition of
  a convex body in the plane, the sum of relative widths of all parts
  is at least one.  For any convex polygon, there exists a direction
  along which an optimal partition is achieved by parallel cuts.
  Given such a direction, an optimal partition can be computed in
  linear time.
\end{abstract}

\section{Introduction} Most works in partitioning polygons have primarily focused on
maximizing geometric measures, such as fatness or the minimum side
length of resulting pieces~\cite{Buchin2021,Damian2004,ORourke2004}.
In this paper, we study an opposite objective: partitioning a polygon
into subpolygons whose widths are bounded above in certain directions.
Such a width constraint commonly arises in manufacturing and recycling
industries, where materials must be cut or processed within certain
width limits.  For example, wood chipping and metal shredding require
pieces to fit within the machine's inlet.  In some materials, cut
directions are critical for preserving structural strength; for
example, fabric is typically cut along the fiber
direction~\cite{Hearle1967}.

Our problem is rooted in classical questions in convex geometry,
notably Tarski's plank problem and its affine-invariant extension by
Bang~\cite{Bang1951}.  The original conjecture asserts that any
covering of a convex body in $\mathbb{R}^d$ by strips must have a
total width at least the minimal width of the body, which was proven
by Bang.  Bang further proposed the \emph{affine plank problem}, in
which strip widths are measured relative to the body's width in the
same direction.  The affine version remains open in general, with only
partial results~\cite{Gardner1988,Akopyan2019,Hunter1993,Ball1991}.
It is known to be equivalent to the Davenport
conjecture~\cite{Alexander1968}, which concerns partitions. Several
partition analogues have been
studied~\cite{Bezdek1996,Akopyan2012,Bezdek2013}.

Our work can also be viewed as a partition analogue of these problems
in the plane, where width constraints replace strips, and simple
polygons substitute convex bodies.

\begin{figure}[t!]
  \centering
  \includegraphics[width=0.75\textwidth]{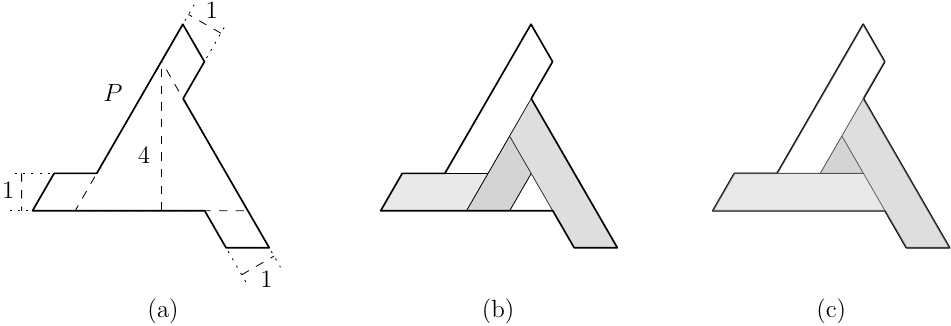}
  \caption{\small (a) The polygon $P$ has a windmill shape with three
    arms extending from an equilateral triangle of height 4.  (b--c)
    The minimum partitions of $P$ under constraints
    $U = \{\vv{v}_{0^\circ}, \vv{v}_{60^\circ}, \vv{v}_{120^\circ}\}$
    and
    $W = \{\vv{v}_{30^\circ}, \vv{v}_{90^\circ},
    \vv{v}_{150^\circ}\}$, where
    $\vv{v}_{\theta} = (\cos \theta, \sin \theta) \in \usetp$.  (b) A
    guillotine partition of five trapezoids. (c) A non-guillotine
    partition of four trapezoids.}
  \label{fig:min.unit.width.partition}
\end{figure}

\subsection{Problem definition and results}
We consider the problem of partitioning polygons into the minimum
number of pieces satisfying both a unit-width constraint and a cut
constraint.  Let $P$ be a simple polygon.  Let $Q$ be a piece in a
partition of $P$ satisfying unit-width constraint $W\subseteq \usetp$,
where $\usetp$ is the set of unit vectors
$\{(\cos \theta, \sin \theta) \mid 0 \le \theta < \pi \}$.  Let
$\dwidth{\vv{v}}{Q}$ denote the \emph{width} of $Q$ in
$\vv{v}\in \usetp$, that is, the length of the orthogonal projection
of $Q$ on a line parallel to $\vv{v}$.  We say $Q$ satisfies the
\emph{unit-width constraint} $W$ if $\dwidth{\vv{v}}{Q}\le 1$ for some
vector $\vv{v} \in W$.

The process of partitioning $P$ must satisfy a cut constraint
$U\subseteq \usetp$.  A \emph{cut} is defined as a line segment within
$P$ whose relative interior lies in the interior of $P$.  A unit
vector is often used to represent a direction, such as the orientation
of a cut, meaning that the cut lies on a line parallel to the vector.
We require that every cut must be in a direction in $U$.  If a cut
extends from one edge of $P$ to the other edge, it divides $P$ into
two distinct pieces; such a cut is called a \emph{guillotine} cut.
 A \emph{guillotine partition} of $P$ (also called a \emph{binary partition}) 
 is obtained by a finite sequence of guillotine cuts; 
 it starts from $P$ and recursively partitions each piece into two subpieces
 using a guillotine cut.  A non-guillotine partition of $P$ is an arbitrary partition of $P$ using cuts that
  are not necessarily guillotine.
  Figure~\ref{fig:min.unit.width.partition} shows a guillotine partition and a non-guillotine partition.

Given a simple polygon $P$, our objective is to partition $P$ into the
minimum number of pieces using cuts constrained by $U$ such that each
piece in the partition satisfies unit-width constraint $W$.  We denote
this minimum partition problem by $\prob{P,W,U}$ and its
  optimum value by $\nopt{P,W,U}$.  Throughout this paper, we use $W$
and $U$ exclusively to refer to the unit-width constraint and the cut
constraint, respectively.

\subparagraph{Related Work.}  Damian and Pemmaraju~\cite{Damian2004b}
and Damian-Iordache~\cite{Damian2000} gave a polynomial-time algorithm
for partitioning a simple polygon into the minimum number of
subpolygons without using Steiner points such that each subpolygon has
diameter at most $\alpha$, for $\alpha>0$.  Later, Buchin and
Selbach~\cite{Buchin2021} showed that this problem becomes NP-hard for
polygons with holes.  Worman~\cite{Worman2003} proved NP-completeness
for a variant in which each subpolygon must be contained in an
axis-aligned square of side length $\alpha$.  Abrahamsen and
Stade~\cite{Abrahamsen2024} showed that allowing Steiner points leads
to NP-hardness for the partition problem under axis-aligned
unit-square containment, even for simple polygons without holes.  This
marks the first known NP-hardness result for a minimum partition
problem on hole-free polygons.

Abrahamsen and Rasmussen~\cite{Abrahamsen2025} studied the problem of
partitioning simple polygons into the minimum number of pieces such
that each piece satisfies a bounded-size constraint (e.g., unit area,
perimeter, diameter, or containment within unit disks or squares).
They noted that computing optimal partitions is hard, even when the
input polygon is a square.

\subparagraph{Our Results.}  Our main contribution is an analysis of
the minimum partition number under constraints
$W, U \subseteq \usetp$.  First, we provide necessary and sufficient
conditions for the existence of feasible partitions in $\prob{P,W,U}$,
along with a decision algorithm for testing feasibility

Second, we study the monotonicity of the minimum partition number
under polygon containment $Q \subseteq P$
(Sections~\ref{sec.monotonicity.partition.num},
\ref{sec:reconfig.restricted.partition},
and~\ref{sec.analysis.reconfiguration}).  We show that this
monotonicity does not hold in general, and identify a sufficient
condition based on a restricted-orientation convexity, called
\emph{$\oset$-convexity}, where $\oset \subseteq \usetp$.
Theorem~\ref{thm:monotonicity.containment} states that
$\nopt{Q,W,U} \leq \nopt{P,W,U}$ holds if $Q$ is $U$-convex with
respect to $P$ for guillotine partitions, or $\overline{W}$-convex
with respect to $P$ for non-guillotine partitions, where
$\overline{W}$ is the set of all unit vectors perpendicular to those
in $W$.

Finally, we prove a partition analogue of Bang's conjecture
(Section~\ref{sec:bang.type.theorem}).  The statement of Bang's
conjecture is as follows: if a convex body $K\subset \mathbb{R}^d$ is
covered by strips $H_1,\ldots, H_m$, then
$\sum_{i=1}^m \inf_{\vv{v}\in \usetp}
\frac{\dwidth{\vv{v}}{H_i}}{\dwidth{\vv{v}}{K}} \ge 1$.  Our theorem
replaces strip coverings with arbitrary partitions and extends the
direction set to any subset $W\subseteq \usetp$.

\begin{theorem}[Bang-Type Partition Analogue]\label{thm:bang.type.analogue}
  Let $K\subset \mathbb{R}^2$ be a convex body, and let
  $K_1\cup \cdots \cup K_m = K$ be its arbitrary partition.  Then, for
  any subset $W \subseteq \usetp$, \[ \sum_{i=1}^m \inf_{\vv{v}
      \in W}\frac{\dwidth{\vv{v}}{K_i}}{\dwidth{\vv{v}}{K}} \ge 1.\]
\end{theorem}

To the best of our knowledge, this is the first partition analogue
that allows non-convex pieces.  We also show that, for
\(U \subseteq \overline{W}\), an optimal partition of a convex polygon
can be computed in linear time using equally spaced parallel cuts (See
Corollary~\ref{cor:convex.opt.partition}). 

\section{Preliminaries}
Let $P$ be a simple polygon with $n$ vertices in the plane.  We assume
that the vertices of $P$ are given in a list sorted in
counterclockwise order along its boundary.  A partition of a simple
polygon is a set of connected pieces with pairwise disjoint interiors
whose union equals the polygon. The cardinality of a partition is the
number of its pieces.

For a set $X \subseteq \mathbb{R}^2$, we denote by $\mybd{X}$ the
boundary of $X$, by $\myint{X}$ the interior of $X$, and by $\mycl{X}$
the closure of $X$.  We treat a polygon as the union of its interior
and boundary; $\mycl{P} = P$, $\mybd{P}$ is the boundary, and
$\myint{P}$ is the interior.

For a point $p \in \mathbb{R}^2$, let $\xcoord(p)$ and $\ycoord(p)$ be
its $x$- and $y$-coordinates, respectively.  For any two points $p$
and $q$ in $\mathbb{R}^2$, we use $\overline{pq}$ to denote the line
segment connecting $p$ and $q$ with length $|\overline{pq}|$.  We call
$\overline{pq}$ a cut in $P$ if it lies entirely in the interior of
$P$, excluding its endpoints.  If both endpoints lie on $\mybd{P}$, we
call it a guillotine cut.

The inner product of any two vectors $\vv{u}$ and $\vv{v}$ is denoted
by $\inner{\vv{u}}{\vv{v}}$.  The Euclidean norm of a vector $\vv{v}$
is denoted by $\norm{\vv{v}}$.  A vector with norm 1 is called a
\emph{unit vector}.  We use $\usetp$ to denote the set of unit vectors
$\{(\cos \theta, \sin \theta) \mid 0 \le \theta < \pi \}$.  For a
subset $V \subseteq \usetp$, we define
$\overline{V} = \{\vv{v}\in \usetp \mid \inner{\vv{u}}{\vv{v}} = 0
\text{ for some } \vv{u}\in V\}$.

For a compact set $X \subseteq \mathbb{R}^2$ and a vector
$\vv{v} \in \usetp$, let $\dwidth{\vv{v}}{X}$ denote the length of the
orthogonal projection of $X$ on a line parallel to $\vv{v}$.  We say
$X$ satisfies unit-width constraint $W \subseteq \usetp$ if and only
if $\dwidth{\vv{v}}{X} \le 1$ for some $\vv{v} \in W$.

A \emph{strip} is the region in the plane bounded by two parallel
lines.  The distance between the bounding lines is the \emph{width} of
the strip, and the direction in $\usetp$ orthogonal to the bounding
lines is called the \emph{normal} vector of this strip.  If a strip
has width 1, we call it a \emph{unit} strip.

We use the notation $[m]\coloneqq \{1,2,\ldots,m\}$ for a positive
integer $m$.  For a finite set $A$, we use $|A|$ to denote the
cardinality of $A$ which is the number of its elements.

\section{Existence of partitions for \texorpdfstring{$W$ and $U$}{W and U}}\label{sec.existence.solutions}
There may not exist a partition of a polygon that satisfies 
unit-width constraint $W$ and cut constraint $U$. 
Consider the example that both $W$ and $U$ consist of 
a single unit vector $(1,0)$. 
If the horizontal width of a polygon is strictly larger than 1,
no partition of the polygon by horizontal cuts imposed by $U$ satisfies $W$.
On the other hand, if $U$ contains two or more distinct vectors, 
it is always possible to partition a polygon using guillotine cuts 
such that the width of each piece in the partition is small enough to satisfy $W$ for any width constraint $W$. 
The same holds for non-guillotine partition as any guillotine partition is also a valid non-guillotine partition.
Thus, to determine the existence of a solution, 
it suffices to consider the case that $U$ is a singleton. 
In this case, every cut for partitioning $P$ must be a guillotine cut.

Without loss of generality, consider the case that $U=\{(0,1)\}$.
Consider the vertical decomposition of $P$ by drawing two vertical extensions from every vertex of $P$.
Observe that these extensions are guillotine cuts in $P$.
Let $\mathcal{T}$ denote the set of those trapezoids in the decomposition.
See Figure~\ref{fig:exist.condition}(a).

One can
observe that $\prob{P,W,U}$ has a solution if and only if
$\probt{W,U}$ has a solution for each $T\in \mathcal{T}$.
The following lemma provides the conditions under which a trapezoid 
has a solution.

\begin{figure}[ht!]
  \centering
  \includegraphics[width=.9\textwidth]{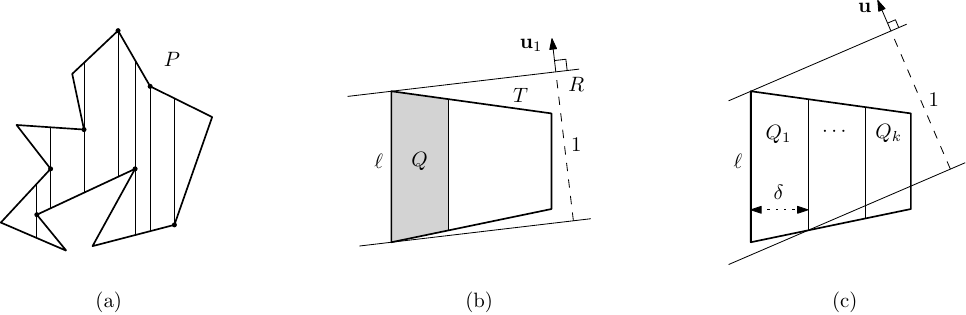}
  \caption{\small 
  (a) $P$ is partitioned into trapezoids by vertical cuts at each vertex of $P$.
  (b) Since $\ell$ is the longer vertical side of a trapezoid $T$, $T\subset R$ for 
  the unit strip $R$ with normal vector $\vv{u}_1$ containing $Q$.
  (c) $\delta$ is a positive value such that $\dwidth{\vv{u}}{Q(\delta)} = 1$, and 
  $Q$ is partitioned into a finite number of subpolygons, each having horizontal width of at most $\delta$.
  }
  \label{fig:exist.condition}
\end{figure}
\begin{lemma}\label{lem:exist.partition}
  Let $W\subseteq \usetp$ be a set of unit vectors. 
  Let $T$ be a trapezoid whose parallel opposite sides are vertical, with $\ell$ being the longer vertical side.
  Then $\probt{W,\{(0,1)\}}$ has a solution if and only if $\dwidth{\vv{u}}{\ell} < 1$ or  
    $\dwidth{\vv{u}}{\ell} = \dwidth{\vv{u}}{T} = 1$ for some $\vv{u} \in W$. 
  \end{lemma}
\begin{proof}
  Suppose that $\probt{W,\{(0,1)\}}$ has a solution, 
  but no vector $\vv{u}\in W$ satisfies $\dwidth{\vv{u}}{\ell} < 1$ or  
    $\dwidth{\vv{u}}{\ell} = \dwidth{\vv{u}}{T} = 1$. 
    Clearly, $\dwidth{\vv{u}}{\ell} \le |\ell|$ by definition of width. 
    If $|\ell| < 1$, then every vector $\vv{u}\in W$ satisfies $\dwidth{\vv{u}}{\ell} < 1$.
    Thus, $|\ell| \ge 1$. 
  
  Since $U=\{(0,1)\}$, every partition of $T$ has a piece $Q$ containing $\ell$. 
  If $\dwidth{\vv{v}}{\ell} > 1$ for every $\vv{v} \in W$, then
  $\dwidth{\vv{v}}{Q}>1$, and thus there is no solution to $\probt{W,\{(0,1)\}}$, a contradiction.
  Therefore, $|\ell| \ge 1$, and there exists $\vv{v} \in W$ such that $\dwidth{\vv{v}}{\ell} \le \dwidth{\vv{v}}{Q} \le 1$.

  Observe that there are exactly two distinct unit vectors $\vv{u}_1,\vv{u}_2 \in \usetp$ 
  for $\ell$ with $|\ell|>1$ 
  such that $\dwidth{\vv{u}_1}{\ell} = \dwidth{\vv{u}_2}{\ell} = 1$. Also, $\vv{u}_1$ and 
  $\vv{u}_2$ are symmetric with respect to the $y$-axis.
  If $|\ell|=1$, then $\vv{u}_1=\vv{u}_2=(0,1)$ is the only
  vector in $\usetp$ satisfying $\dwidth{\vv{u}_1}{\ell} = \dwidth{\vv{u}_2}{\ell} = 1$.
  For any vector $\vv{v} \in W \setminus\{\vv{u}_1, \vv{u}_2\}$, we have
  $\dwidth{\vv{v}}{\ell} > 1$ by the assumption. 
  
  Let $\Pi$ be an optimal partition for $\probt{W,\{(0,1)\}}$,
  and let $Q'$ be a piece in $\Pi$ that contains $\ell$.
  Since $\Pi$ is a solution, 
  $Q'$ satisfies $W$, i.e., $\dwidth{\vv{v}}{Q'} \le 1$ for some $\vv{v} \in W$.
  By the assumption, we have $1 \le \dwidth{\vv{v}}{\ell}$. 
  Moreover, since $\ell \subseteq Q'$, $\dwidth{\vv{v}}{\ell} \le \dwidth{\vv{v}}{Q'}$. 
  Combining these inequalities, we obtain $\dwidth{\vv{v}}{\ell} = 1$, which implies that
  $\vv{v}$ must be $\vv{u}_1$ or $\vv{u}_2$.

  Consider the case $\vv{v} = \vv{u}_1$.
  Let $R$ be the unit strip with normal vector $\vv{u}_1$ that contains $Q'$.  
  Observe that $T$ and $Q'$ share the segment $\ell$ as well as non-empty portions of the top and bottom sides of $T$ 
  that are incident to $\ell$. See Figure~\ref{fig:exist.condition}(b).
  Since $T$ is a trapezoid, $T$ is also contained in $R$.
  Thus, $\dwidth{\vv{u}_1}{T} = 1$, a contradiction. Similarly, we have a contradiction
  for the case $\vv{v} = \vv{u}_2$.
  Therefore, there is a unit vector $\vv{u} \in W$ satisfying one of the conditions.

  Conversely, assume that there is a unit vector $\vv{u} \in W$ satisfying $\dwidth{\vv{u}}{\ell} < 1$ or  
    $\dwidth{\vv{u}}{\ell} = \dwidth{\vv{u}}{T} = 1$.
  Clearly, $\nopt{T,W,\{(0,1)\}} = 1$ if $\dwidth{\vv{u}}{T} = 1$.   
  We consider the case that $\dwidth{\vv{u}}{\ell} < 1$ and $\dwidth{\vv{u}}{T} > 1$.
  Without loss of generality, we assume that $\ell$ lies on $x=t$ and appears on the left side of $T$.
  Let $\delta>0$ be the largest value such that the subpolygon $Q$ in the partition of $T$ 
  induced by $x=t+\delta$ and lying to the left of $x=t+\delta$ satisfies $\dwidth{\vv{u}}{Q} \le 1$.
  Let $\Pi=\{Q_1,Q_2, \ldots, Q_k\}$ be the partition of $T$ obtained by drawing lines at $x=t+i\delta$ 
  for $i=1,2,\ldots$, 
  until the rightmost subpolygon $Q_k$ 
  has horizontal width at most $\delta$. 
  Observe that each $Q_j$ has its top and bottom edges 
  lying on the top and bottom edges of $T$, respectively. 
  Since $\ell$ is the longer vertical side of $T$, 
  every $Q_j$ for $j=2,\ldots,k$ can be translated to lie inside $Q_1$. 
  It follows that $\dwidth{\vv{u}}{Q_j}\le \dwidth{\vv{u}}{Q_1} \le 1$ for every $j = 2,\ldots, k$.
  Thus, $\probt{W,\{(0,1)\}}$ has a solution.
\end{proof}

In summary, $\prob{P,W,U}$ has a solution if $U$ contains two distinct vectors in $\usetp$. 
It can be done in $O(1)$ time 
by checking whether $U$ is a singleton, since $U \subseteq \usetp$.
For $U$ consisting of one vector in $\usetp$, 
we first perform a trapezoidal decomposition of $P$, and then 
check whether each trapezoid satisfies the condition stated in Lemma~\ref{lem:exist.partition}.
By employing a trapezoidal decomposition technique~\cite{Chazelle1990}, 
we can determine whether $\prob{P,W,U}$ has a solution in $O(n+|W|)$ time
for any finite set $W \subseteq \usetp$.

\begin{corollary}\label{cor:existential.test}
  We can determine whether a simple $n$-gon $P$ has a solution to $\prob{P,W,U}$
  for any finite sets $W,U\subseteq \usetp$ in $O(n + |W|)$ time.  
\end{corollary}
\begin{proof}
  Consider the case that $U$ consists of one vector in $\usetp$. Without loss of generality, let  $U=\{(0,1)\}$.
  The trapezoidal decomposition of $P$ can be done in $O(n)$ time~\cite{Chazelle1990}. 
  We find the vector $\vv{v} \in W$ that maximizes $|\inner{\vv{v}}{(0,1)}|$ in $O(|W|)$ time. 
  This corresponds to the direction that minimizes $\dwidth{\vv{v}}{\ell}$ for any vertical segment $\ell$. 
  For each trapezoid $T$, we test in $O(1)$ time whether the selected vector $\vv{v}$ satisfies the condition. 
  If $\dwidth{\vv{v}}{\ell} \neq 1$, we can directly determine whether $\prob{T,W,U}$ has a solution. 
  In the case that $\dwidth{\vv{v}}{\ell}=1$,
  $\dwidth{\vv{v}}{T} = 1$ if and only if its top and bottom sides lie in the unit strip $R$ with normal vector $\vv{v}$ that contains $\ell$.
  Since the number of trapezoids is $O(n)$, the total running time is $O(n + |W|)$.
\end{proof}

\section{Monotonicity of minimum partition numbers}\label{sec.monotonicity.partition.num}
In this section, we assume that $\prob{P,W,U}$ has a partition
satisfying the constraints.  
The necessary and sufficient condition for feasibility is presented in 
Section~\ref{sec.existence.solutions}.
We show $\nopt{Q,W,U}\le \nopt{P,W,U}$ for any subpolygon $Q$ of $P$
that satisfies a certain condition.  For both guillotine and
non-guillotine partitions, we identify sufficient conditions on $Q$
that ensure this monotonicity.  When the constraints $W$ and $U$ are
clear from context, we abbreviate $\nopt{P,W,U}$ as $\nopt{P}$.

\begin{figure}[ht!]
  \centering
  \includegraphics[width=0.9\textwidth]{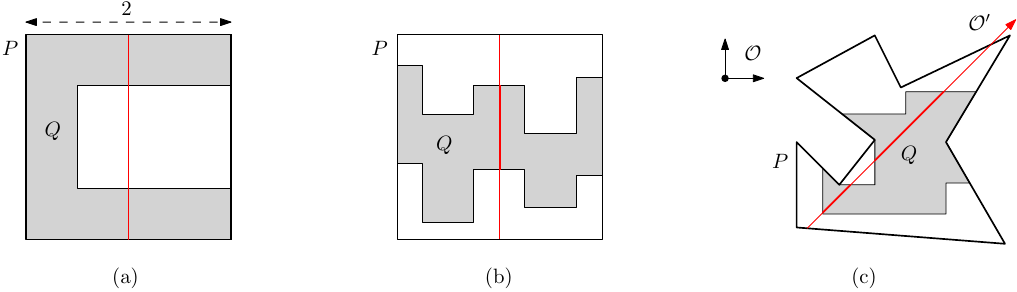}
  \caption{\small (a) A polygon $P$ and its subpolygon $Q$ (gray) with
    $\nopt{P,W,U} = 2$ and $\nopt{Q,W,U} > 2$ for $W = \{(1,0)\}$ and
    $U = \{(0,1)\}$.  (b) $Q$ (gray) is $U$-convex (or
    $\overline{W}$-convex) with respect to $P$, resulting in
    $\nopt{Q,W,U} \le \nopt{P,W,U}$.  (c) $Q$ (gray) of $P$ is
    $\oset$-convex, but not $\oset'$-convex with respect to $P$, where
    $\oset = \{(1,0),(0,1)\}$ and
    $\oset'=\{(\frac{\sqrt{2}}{2},\frac{\sqrt{2}}{2})\}$.}
  \label{fig:restricted.monotonicity}
\end{figure}

Let $P$ be a square of side length 2, and let $Q$ be a subpolygon of
$P$ as shown in Figure~\ref{fig:restricted.monotonicity}(a).  Consider
the instance $\prob{P,W,U}$ with $W = \{(1,0)\}$ and $U = \{(0,1)\}$.
Since $U$ is a singleton, every cut must be a guillotine cut.  Observe
that neither $P$ nor $Q$ satisfies unit-width constraint $W$.  A
vertical cut halving $P$ yields a feasible partition of two pieces.
No vertical cut, however, in $Q$ partitions $Q$ into two pieces, each
with horizontal width at most 1.  Thus, $\nopt{Q} >2 = \nopt{P}$,
implying that the inclusion $Q \subseteq P$ alone is not sufficient to
ensure $\nopt{Q}\le \nopt{P}$.

Such failures are common in minimum partitioning problems, as the
optimal partition number depends on the geometric complexity of $Q$.
Two typical approaches are restricting the geometry of $Q$ relative to
$P$, and constraining the partition class to specific families, such
as the guillotine (binary) class~\cite{Bezdek1996}.  These alter the
structural properties of feasible solutions.

\subsection{Restricted-orientation convexity}
Rawlins~\cite{Rawlins1987} introduced the \emph{restricted-orientation
  convexity} as a generalization of standard convexity in Euclidean
space.  For $\oset \subseteq \usetp$, a set $X$ is
\emph{$\oset$-convex} if, for every line $\ell$ parallel to a vector
in $\oset$, $X\cap \ell$ is connected (we regard an empty set as being
connected).  When $\oset = \usetp$, $\oset$-convexity coincides with
standard convexity in $\mathbb{R}^2$.  We extend this concept to
subpolygons of simple polygons with respect to guillotine cuts.
\begin{definition}
  Let $Q$ be a subpolygon of a polygon $P$, and $\oset$ be a set of
  unit vectors.  Then $Q$ is $\oset$-convex with respect to $P$ if its
  intersection with every guillotine cut in $P$ parallel to a vector
  in $\oset$ is connected.
\end{definition}
Figure~\ref{fig:restricted.monotonicity}(c) shows a subpolygon $Q$
that is $\oset$-convex with respect to $P$ for
$\oset=\{(0,1), (1,0)\}$, but not $\oset'$-convex for
$\oset'=\{(\frac{\sqrt{2}}{2},\frac{\sqrt{2}}{2})\}$ since a
guillotine cut in $P$ parallel to
$(\frac{\sqrt{2}}{2},\frac{\sqrt{2}}{2})$ intersects $Q$ in at least
two connected components.

Observe that $\oset$-convexity with respect to a polygon $P$ is
equivalent to the \emph{geodesic convexity} within $P$ when
$\oset = \usetp$.
The following theorem gives sufficient conditions for monotonicity of
the minimum partition numbers.
\begin{theorem}[Monotonicity in Polygon
  Containment]\label{thm:monotonicity.containment}
  Let $U$ and $W$ be sets of unit vectors, and let $Q$ be a subpolygon
  of a polygon $P$.  Assume that $\prob{P,W,U}$ has a solution.
  \begin{itemize}\denseitems
  \item $\nopt{Q} \le \nopt{P}$ for guillotine partitions if $Q$ is
    $U$-convex with respect to $P$.
  \item $\nopt{Q} \le \nopt{P}$ for non-guillotine partitions if $Q$
    is $\overline{W}$-convex with respect to $P$.
  \end{itemize}
\end{theorem}

Revisit the subpolygons $Q$ in
Figure~\ref{fig:restricted.monotonicity}(a,b).  In (a), $Q$ is neither
$U$-convex nor $\overline{W}$-convex with respect to $P$, whereas in
(b), $Q$ satisfies both conditions, and
Theorem~\ref{thm:monotonicity.containment} implies that
$\nopt{Q} \le \nopt{P} = 2$, for both guillotine and non-guillotine
partitions.

\begin{figure}[t!]
  \centering
  \includegraphics[width=0.65\textwidth]{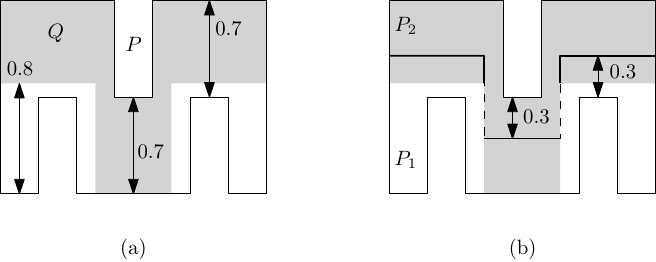}
  \caption{\small A non-guillotine partition of $P$, where
    $W=\{(0,1)\}$ and $U =\{(0,1),(1,0)\}$.  (a) The polygon $P$ and a
    subpolygon $Q$, both having vertical widths greater than 1.  (b)
    An optimal non-guillotine partition of $P$ into two pieces, but
    four when restricted to $Q$.}
  \label{fig:restricted.partition}
\end{figure}

\subsection{Monotonicity in guillotine partitions}\label{sec:monotonicity.guillotine}
We prove Theorem~\ref{thm:monotonicity.containment} for the
  guillotine case. Let $\Pi = \{P_1, \ldots, P_m\}$ be a guillotine
partition of $P$ feasible to $\prob{P,W,U}$, and let
  $\Pi[Q] = \{P_i \cap Q\}_{i=1,\ldots,m}$ be its restriction to a
  subpolygon $Q$ that is $U$-convex with respect to $P$. Then,
$\bigcup_{i=1}^m (P_i \cap Q) = Q$, and thus, $\Pi[Q]$ also forms a
partition of $Q$.  Note that each piece in $\Pi[Q]$ satisfies
unit-width constraint $W$, as every piece is a subset of some $P_i$
that satisfies the constraint.  To show $\nopt{Q} \le \nopt{P}$, it
suffices to verify two aspects: (1) $\Pi[Q]$ is a guillotine partition
of $Q$ and (2) $|\Pi[Q]|\le m$.

Since $\Pi$ is a guillotine partition, the cuts used for $\Pi$ are
partially ordered by their precedence in the partitioning process.
Let $\mathcal{C}=\langle c_1,\ldots,c_{m-1}\rangle$ denote a sequence
of guillotine cuts that produces $\Pi$.  Then each $c_i \cap Q$ is
connected as $c_i$ is a cut in $P$ with a direction in $U$ and $Q$ is
$U$-convex with respect to $P$.  Let
$\mathcal{D}=\langle d_1,\ldots, d_{m'}\rangle$ be the
  subsequence of $\langle c_1 \cap {Q}, \ldots, c_{m-1}\cap {Q} \rangle$ 
  consisting of those satisfying $c_i \cap \myint{Q}\neq\emptyset$ for $i=1,\ldots,m-1$. Note
that each $d_j$ for $j\in[m']$ is a cut in $Q$ and $m'\le m-1$.  Also,
observe that $\Pi[Q]$ is induced by $\mathcal{D}$.  It suffices to
show that $\mathcal{D}$ is indeed a sequence of guillotine cuts in
$Q$: each $d_j\in \mathcal{D}$ is a guillotine cut in one subpolygon
obtained by applying $\langle d_1, \ldots, d_{j-1} \rangle$ to $Q$.

Let $i$ be the smallest index in $[m-1]$ such that
$c_{i} \cap \myint{Q}\neq\emptyset$.  Since no cut in
$\langle c_1,\ldots, c_{i-1}\rangle$ intersects $\myint{Q}$, there is
a subpolygon in the partition of $P$ by
$\langle c_1,\ldots, c_{i_1 - 1}\rangle$ that contains $Q$.  Observe
that $c_{i}$ is a guillotine cut in this subpolygon that is aligned
with a direction in $U$.  Since $Q$ is $U$-convex with respect to $P$,
\ $d_1=c_{i} \cap Q$ is a line segment, and thus, $d_1$ is a
guillotine cut in $Q$.

We proceed by induction on $j$ with $1 < j \le m'$, that
$\langle d_1, \ldots, d_{j-1}\rangle$ forms a sequence of guillotine
cuts in $Q$.  Among the pieces of $Q$ partitioned by the sequence, let
$Q_j$ denote the one containing $d_j$, where $d_j$ is a cut in $Q_j$.
By definition, there exists an index $k\in [m]$ such that
$d_j = c_{k} \cap Q$.  Let $P_k$ be the piece in the partition of $P$
by $\langle c_1,\ldots, c_{k-1}\rangle$ that contains $c_k$.

To see that $d_j$ is a guillotine cut of $Q_j$, recall that
  $d_j = c_k \cap Q_j$, where $c_k$ is a guillotine cut of $P_k$ and
  $Q_j \subseteq P_k$.  Let $c_k'$ be the guillotine cut in $P$
  obtained by extending $c_k$ until it touches $\mybd{P}$.  Since $Q$
  is $U$-convex, $c_k'\cap Q$ is connected, and hence so is its
  subsegment $c_k \cap Q = d_j$.  Moreover, as $c_k$ spans
  $\mybd{P_k}$, its restriction to $Q_j$ necessarily touches
  $\mybd{Q_j}$ at both endpoints.  Thus, $d_j$ is a guillotine cut in
  $Q_j$. We conclude that
$\mathcal{D}=\langle d_1, \ldots, d_{m'}\rangle$ induces a guillotine
partition of $Q$. Since $m'\le m$, we have $|\Pi[Q]|\le m$.

\begin{lemma}\label{lem:guillotine.monotonicity}
  Let $P$ be a polygon and let $\Pi = \{P_1,\ldots,P_m\}$ be a
  solution to the problem $\prob{P,W,U}$ for guillotine cuts.  If a
  subpolygon $Q$ of $P$ is $U$-convex with respect to $P$, the
  restricted partition $\{Q\cap P_i\}_{i = 1, \ldots, m}$ is a
  solution to $\prob{Q,W,U}$ for guillotine cuts with at most $m$
  pieces.
\end{lemma}

\section{Reconfiguration of restricted non-guillotine partitions}\label{sec:reconfig.restricted.partition}
Let $Q$ be a subpolygon of $P$ that is $\overline{W}$-convex with
respect to $P$.  The monotonicity $\nopt{Q} \le \nopt{P}$ trivially
holds when $\nopt{Q} = 1$.  Also, any feasible partition is guillotine
when $U$ is a singleton.  Assume that $\nopt{Q} > 1$ and $U$ contains
at least two distinct vectors.

Let $\Pi = \{P_1,\ldots, P_m\}$ be any feasible partition of $P$ to
$\prob{P,W,U}$.  Its restriction to $Q$, defined as
$\Pi[Q] =\{P_i \cap Q\}_{i=1,\ldots,m}$, is a feasible partition of
$Q$ to $\prob{Q,W,U}$.  However, some regions $P_i \cap Q$ may be
disconnected, even when $Q$ is $\overline{W}$-convex with respect to
$P$ (See Figure~\ref{fig:restricted.partition}(a--b)).  To address
this, we modify the cuts in $\Pi[Q]$ to reconnect disjoint fragments
into connected regions while preserving feasibility for
$\prob{Q, W, U}$.

Consider any element $P_i$ in the partition $\Pi$ of $P$ such that
$\myint{P_i} \cap \myint{Q} \neq \emptyset$.  The intersection
$\myint{P_i} \cap \myint{Q}$ consists of open connected components.
We define $R_i$ as \[R_i \coloneqq \left\{ \, \mycl{X}
    \,\middle |\, X \text{ is a connected component of } \myint{P_i}
    \cap \myint{Q} \text{ with } \mycl{X} \cap \mybd{Q} \neq \emptyset
    \, \right\}.\]

Let $R_i = \{C_1, C_2, \ldots, C_t\}$.  Note that each $C_j$ is a
subpolygon of $Q$ with positive area that touches $\mybd{Q}$.  To
connect elements in $R_i$ into a single piece, we reconfigure $\Pi[Q]$
by iteratively performing a process, called \emph{reallocation}.  Let
$X$ be a connected region in $Q$.  We define the reallocation of $X$
to $P_i$ as the operation that modifies $\Pi[Q]$ by expanding the
region assigned to $P_i$ so that it includes $X$.  This can be done by
adding new boundaries along $\mybd{X} \setminus P_i$ and removing
those along $\mybd{P_i} \setminus X$.  If $X$ and $P_i$ intersect in
their interiors, or share a boundary segment of positive length, then
$X \cup P_i$ appears as a single piece in the resulting partition.  We
say that the region $X$ is \emph{reallocated} to $P_i$.

\subsection{\texorpdfstring{Construction and layering of corridor of $Q$}
{ Construction and Layering of Corridor of Q}}
We first construct a narrow corridor along $\mybd{Q}$.  This corridor
lies entirely within $Q$, closely following $\mybd{Q}$, and provides
the space needed to link elements of $R_i$.

We borrow the definitions from~\cite{Barequet2014}.
For $\delta \ge 0$, the inner $\delta$-annulus of 
the polygon $Q$ is the closed region containing all points inside $Q$ 
at distance at most $\delta$ from $\mybd{Q}$.
The inner \emph{$\delta$-offset} polygon 
of $Q$, denoted by $Q^\delta$, is defined as the polygon obtained from 
the inner $\delta$-annulus of $Q$ 
by extending each of the line segments until they intersect with other extensions. 
Note that the inner $0$-offset polygon of $Q$ is $Q$ itself.
Figure~\ref{fig:offset.polygon}(a--b) illustrates the inner $\delta$-annulus and $\delta$-offset polygon of $Q$.

Observe that these extensions may intersect outside $Q$ so that the offset polygon is not simple.  
A more formal definition of the inner $\delta$-offset polygon 
can be derived using the shrinking process in~\cite{Aichholzer1996}.
For sufficiently small $\delta\ge 0$, however,
$Q^\delta$ is well-defined as a simple polygon entirely contained within $Q$. 
Each edge of the offset polygon 
is parallel to its corresponding edge of $Q$, and appears in the same cyclic order 
along the boundary of the offset polygon as in $Q$. 

The $\delta$-\emph{corridor} of $Q$ is 
defined as the closed region that lies between $Q$ and $Q^\delta$.
This $\delta$-corridor is topologically equivalent to an annulus, allowing us to define its inner boundary and outer boundary.
The inner boundary of the $\delta$-corridor coincides with $\mybd{Q^\delta}$, 
while its outer boundary is $\mybd{Q}$.
The parameter $\delta$ is referred as the width of this corridor.

\subparagraph{Determining width of the Corridor.} 
The width of this corridor is carefully chosen to ensure that 
no reallocation violates the width constraint $W$.

For each $C_j \in R_i$, 
define $\bdeset_{ij}$ as the set of edges of $C_j$ 
that both intersect $\myint{Q}$ and have an endpoint on $\mybd{Q}$. We then set
\[
\bdeset_i := \bigcup_{j=1}^{t} \bdeset_{ij}, \qquad \bdeset_Q := \bigcup_{i=1}^{m} \bdeset_i.
\]


As $\delta$ increases from 0, 
the offset boundary $\mybd{Q^\delta}$ shrinks inward from $\mybd{Q}$. 
By the combinatorial structure of $\mybd{Q^\delta} \cap \mybd{C_j}$,
we mean the set of 
incidence pairs 
specifying which edge or vertex of $Q^\delta$ meets which edge or vertex of $C_j$.
This structure changes only at discrete critical values when 
(1) an edge of $Q^\delta$ contains a vertex of $C_j$, or 
(2) a vertex of $Q^\delta$ lies on the interior of an edge of $C_j$.

We define $\phi_{ij}>0$ as 
the smallest value of $\delta$ 
at which the combinatorial structure of $\mybd{Q^\delta} \cap \mybd{C_j}$ changes. 
We then set $\phi_i \coloneqq \min_{j \in [t]}\phi_{ij}$ and $\phi \coloneqq \min_{i \in [m]}\phi_i$.
See Figure~\ref{fig:offset.polygon}(c).

Finally, we define the width of the corridor to be $\phi$, and denote this corridor by $\corridor$.
By construction, $\bdeset_Q$ consists of those edges intersecting $\myint{\corridor}$.



\begin{figure}[t!]
  \centering
  \includegraphics[width=0.85\textwidth]{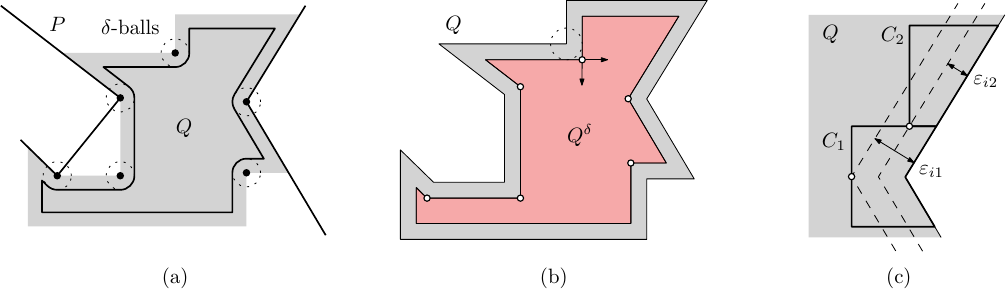}
  \caption{\small (a) The inner $\delta$-annulus of $Q$ is an annular region in $Q$ 
  at distance at most $\delta$ from $\mybd{Q}$. (b) The red colored region $Q^\delta$ is the inner $\delta$-offset polygon of $Q$. 
  (c) For each $C_j\in R_i$, 
  $\phi_{ij}$ is defined by the smallest value of $\delta$ 
  at which $\mybd{Q^\delta}$ touches a new vertex or edge of $C_j$ as $\delta$ increases from 0.}
  \label{fig:offset.polygon}
\end{figure}

\subparagraph{Construction of Layers Within $\corridor$.}
The corridor
$\corridor$ is subdivided into a sequence of nested subregions, which
we refer to as \emph{layers}.  
Let $\nlayers \in \mathbb{Z}_{>0}$
denote the number of layers which will be determined in
Section~\ref{subsubsec.layer.assignment}.  

For an integer $k$ with $0\le k \le \nlayers$, 
let $B_k$ denote the boundary of
the inner $\delta_k$-offset polygon of $Q$, 
where $\delta_k= \frac{\phi \cdot k}{\nlayers}$. 
The boundaries $B_0, \ldots, B_{h}$ are simple closed curves that are pairwise disjoint. 
Recall that $U$ contains two distinct vectors. 
Using these two directions, for each $k = 1, \ldots, \nlayers$,
we construct a simple closed curve $\Gamma_k$ in a zigzag pattern within 
the annular region between $B_{k-1}$ and $B_{k}$. 
This curve $\Gamma_k$ is designed to avoid touching $B_{k-1}$ and $B_{k}$ 
while intersecting each edge in $\bdeset_Q$. 

Let $\delta_{k-1}'$ and $\delta_{k}'$ be two values in the open interval 
$(\delta_{k-1}, \delta_{k})$ satisfying $\delta_{k-1}' < \delta_{k}'$.
Let $B_{k-1}'$ and $B_{k}'$ denote the boundaries of 
$\delta_{k-1}'$-offset polygon 
and $\delta_{k}'$-offset polygon of $Q$, respectively.
The boundaries $B_{k-1}, B'_{k-1}, B'_{k}$ and $B_{k}$ are pairwise disjoint and 
arranged sequentially inward from $\mybd{Q}$. 

By the definition of $\bdeset_Q$, 
every edge $e\in \bdeset_Q$ is a line segment 
that touches a point on $\mybd{Q}$ and 
intersects the inner boundary of $\corridor$, 
as $\phi$ is set to be sufficiently small.
Then the annular region between $B'_{k-1}$ and $B'_{k}$ is divided into multiple subpolygons by the edges in $\bdeset_Q$. 
Each subpolygon is simple and has positive area. 
Exactly two of its edges are subsegments of edges in $\bdeset_Q$; the others lie along the boundaries $B'_{k-1}$ and $B'_{k}$.
We triangulate each subpolygon 
by adding new edges that connect $B'_{k-1}$ to $B'_{k}$; 
the union of these forms a triangulation of the entire annular region.
For each edge in the triangulation 
that connects $B'_{k-1}$ to $B'_{k}$, 
we mark a point on its interior. 
The curve $\Gamma_k$ is then constructed to pass through these marked points. 
Figure~\ref{fig:construction.corridor}(a) shows the triangulation of the region restricted to $C_j$
with its marked points.

\begin{figure}[t!]
  \centering
  \includegraphics[width=0.85\textwidth]{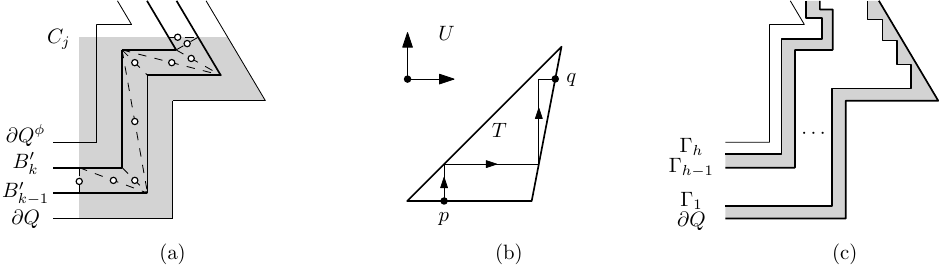}
  \caption{\small (a) The intersection of $C_j$ with the region between $B'_{k-1}$ and $B'_{k}$ 
  is triangulated by adding the edges connecting $B'_{k-1}$ to $B'_{k}$. 
  Each edge connecting $B'_{k-1}$ to $B'_{k}$ is marked at its interior point.
  (b) The simple path from $p$ to $q$ lies in $T$, 
  and each segment is aligned with a direction in $U$. 
  (c) The region between $\mybd{Q}$ and $\Gamma_1$ defines the first layer,  
  and the region between $\Gamma_{\nlayers-1}$ 
  and $\Gamma_{\nlayers}$ defines the $\nlayers$-th layer.
  }
  \label{fig:construction.corridor}
\end{figure}
\begin{lemma}\label{lem:exits.curve.quadrilateral}
  Let $T$ be a triangle with a positive area, and 
  let $p$ and $q$ be two points lying on the interiors of edges of $T$. 
  Assume that $U$ contains two distinct vectors.
  Then, there exists a simple path within $T$ from $p$ to $q$, 
  consisting of line segments 
  aligned with the directions in $U$.
\end{lemma}
\begin{proof}
  Let $\vv{u}$ and $\vv{v}$ be two distinct vectors in $U$. 
  Since they are not parallel, 
  there is an invertible linear transformation $f\colon \mathbb{R}^2 \to \mathbb{R}^2$ such that 
  $f(\vv{u})= (1,0)$ and $f(\vv{v}) = (0,1)$. 
  Note that $f$ transforms $T$ into a triangle $T'$ with a positive area in $\mathbb{R}^2$, 
  and $f(p)$ and $f(q)$ are points on the interiors of edges of $T'$. 
  Thus, $f$ provides a one-to-one correspondence between simple paths in $T$ aligned with $\{\vv{u},\vv{v}\}$ and 
  rectilinear simple paths in $T'$.
  We assume without loss of generality that $\{(1,0),(0,1)\} \subseteq U$, and 
  we find the rectilinear simple path within $T$ connecting $p$ to $q$.

  Consider the case that $x(p) \le x(q)$ and $y(p)\le y(q)$. 
  Starting from $p$, the rectilinear simple path
  alternates between moving rightward and upward. 
  It proceeds horizontally to the right until it meets an edge of $T$, 
  then turns upward and continues vertically 
  until it encounters another edge. 
  This alternating process is repeated.
  If this process continues indefinitely,
  the path either reaches the uppermost vertex of $T$ in finite steps,
  or asymptotically converges to it without ever reaching it.
  Since $q$ is lower than the uppermost vertex, 
  there is a positive integer $t$ such that after $t$ iterations of this process,
  the simple path reaches $z \in \myint{T}$ with $y(z) = y(q)$. 
  By moving horizontally at $z$, the path finally reaches $q$. 
  See Figure~\ref{fig:construction.corridor}(b).
\end{proof}

By Lemma~\ref{lem:exits.curve.quadrilateral}, 
we construct a simple path within each triangle 
that passes the marked points on its edges. 
The curve $\Gamma_k$ is formed by joining these simple paths into a closed simple loop.
This loop consists of line segments aligned with the directions in $U$
and does not touch either $B_{k-1}$ or $B_{k}$, since it is fully contained in the annular region 
between $B_{k-1}'$ and $B_{k}'$.

The curves $\Gamma_1, \ldots, \Gamma_\nlayers$ subdivide $\corridor$ into $\nlayers$ nested
subregions, called \emph{layers}: for each $k \in [\nlayers]$, 
we denote by $L_k$ 
the $k$-th layer, i.e., the annular region bounded by
$\Gamma_{k-1}$ and $\Gamma_k$, where we define $\Gamma_0 = \mybd{Q}$. 
See Figure~\ref{fig:construction.corridor}(c).

\begin{figure}[t!]
  \centering
  \includegraphics[width=0.9\textwidth]{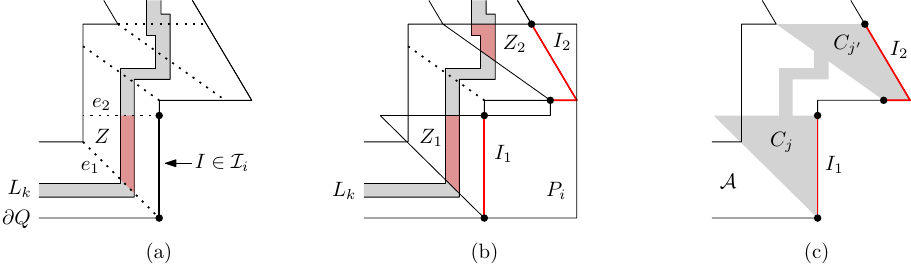}
  \caption{\small (a) A layer segment $Z$ of $R_i$ is bounded by
    $e_1,e_2 \in \bdeset_i$, and it corresponds to
    $I \in \mathcal{I}_i$.  (b) A link instruction
    $\lambda = (I_1,I_2,k)$ reallocates the layer segments in
    $L_k[Z_1,Z_2]$ to $P_i$.  (c) Applying $\lambda$ merges $C_{j}$
    and $C_{j'}$ into a single piece through layer segments in
      $L_k$, where $I_1 \in \mathcal{I}_{ij}$ and
    $I_2 \in \mathcal{I}_{ij'}$ for some $j,j' \in [t]$.}
  \label{fig:reallocation}
\end{figure}

\subsection{Link instructions with circular intervals and layers}
For each edge $e \in \bdeset_Q$, 
the intersection $\corridor \cap e$ consists of one or two line segments, 
depending on whether $e$ is a guillotine cut in $Q$. 
In the construction of layers, we place marked points on each of these segments, 
and each curve $\Gamma_k$ crosses each segment exactly once, at one of the marked points.  
These crossings subdivide each layer into smaller regions, 
each corresponding to a contiguous portion of $\mybd{Q}$, 
later formalized as a circular interval.

Each connected component of $L_k \cap C_j$ for $C_j \in R_i$ is a subregion of $Q$ 
referred to as a \emph{layer segment} of $C_j$ (and of $R_i$) in $L_k$.
Each layer segment is bounded by four parts: continuous
portions of the inner and outer boundaries of $L_k$, and two edges
from $\bdeset_{ij}$.  Removing any layer segment from $L_k$ alters its
topological structure from an annulus to a weakly simple polygon.  All
layer segments in $L_k$ can be arranged in cyclic order along $L_k$.
Figure~\ref{fig:reallocation}(a) illustrates a layer segment of $R_i$
in $L_k$.

We use an interval notation over layer segments that are cyclically
ordered.  For two layer segments $Z_1, Z_2$ in $L_k$, we denote by
$L_k[Z_1,Z_2]$ the set of layer segments encountered in
counterclockwise traversal from $Z_1$ to $Z_2$ in $L_k$, where
$Z_1, Z_2 \in L_k[Z_1,Z_2]$.  Similarly, $L_k(Z_1, Z_2)$,
$L_k[Z_1, Z_2)$, and $L_k(Z_1, Z_2]$ denote the open and half-open
intervals over layer segments between $Z_1$ and $Z_2$ in $L_k$.

Assume that $Z_1 \subseteq C_{j}$ and $Z_2 \subseteq C_{j'}$ are layer
segments of $R_i$ in $L_k$, for $j,j' \in [t]$ with $j \neq j'$.  The
reallocation of $L_k[Z_1, Z_2]$ to $P_i$ merges $C_{j}$ and $C_{j'}$
into a single piece within $L_k$.  We refer to such an ordered pair
$(Z_1,Z_2)$ as a \emph{link instruction} of $R_i$.  See
Figure~\ref{fig:reallocation}(b--c).

\subparagraph{Decomposition of $\mybd{Q}$ into circular intervals.}
For $C_j \in R_i$, $C_j$ is a subpolygon of $Q$ that touches
$\mybd{Q}$.  The intersection $\mybd{C_j} \cap \mybd{Q}$ consists of
continuous paths on $\mybd{Q}$.  Specifically, it can be the loop
$\mybd{Q}$ itself if $\mybd{Q} \subseteq \mybd{C_j}$, implying that
$Q$ satisfies unit-width constraint $W$.  However, assuming
$\nopt{Q} > 1$, this case does not occur.  Thus, every connected
component of $\mybd{C_j} \cap \mybd{Q}$ must be a continuous path on
$\mybd{Q}$ that is not a loop.  We define $\mathcal{I}_{ij}$ as
\[\mathcal{I}_{ij} \coloneqq \left\{ I \subseteq \mybd{Q}
    \;\middle|\; I \text{ is a connected component of } \mybd{C_j}
    \cap \mybd{Q} \right\}.\]
We then define the aggregate sets
$\mathcal{I}_i \coloneqq \bigcup_{j=1}^t \mathcal{I}_{ij}$ and
$\mathcal{I}_{Q} \coloneqq \bigcup_{i=1}^m \mathcal{I}_i$.

Let $I \in \mathcal{I}_{ij}$ be a continuous path along $\mybd{Q}$
with endpoints $p$ and $q$ such that $I$ corresponds to the portion of
$\mybd{Q}$ from $p$ to $q$ in counterclockwise order.  Since
$\mybd{Q}$ is a simple closed curve, topologically equivalent to a
circle, we regard each path $I$ as a circular interval on $\mybd{Q}$.
As $\mathcal{I}_{ij}$ contains no loops, the case $p=q$ occurs only
when $\mybd{C_j}$ touches $\mybd{Q}$ at a single point $p$, and no
other portion of $\mybd{C_j}$ near $p$ intersects $\mybd{Q}$.  Such
intervals are called \emph{degenerate} intervals.

For $i \in [m]$, $\mathcal{I}_i$ may contain both degenerate and
non-degenerate intervals.  Since $P_i$ is simple, all intervals in
$\mathcal{I}_i$ are pairwise interior-disjoint.
Figure~\ref{fig:traversal.pieces}(a) shows a piece $P_i$ and the
subpolygon $Q$, from which the circular intervals in $\mathcal{I}_i$
are defined.  These intervals are illustrated in
Figure~\ref{fig:traversal.pieces}(b).

Note that a degenerate interval at some point $p\in \mybd{Q}$ may
appear in multiple $\mathcal{I}_i$'s.  In such cases, each occurrence
is treated as a distinct element in $\mathcal{I}_{Q}$.  Since
$\Pi = \{P_1,\ldots,P_m\}$ is a partition of $P$ and $Q\subseteq P$,
every point on $\mybd{Q}$ belongs to some interval of
$\mathcal{I}_{Q}$, and no two intervals in $\mathcal{I}_Q$ intersect
each other in their interiors.  Thus, $\mathcal{I}_{Q}$ forms a
decomposition of $\mybd{Q}$.

\begin{figure}[t]
  \centering
  \includegraphics[width=0.85\textwidth]{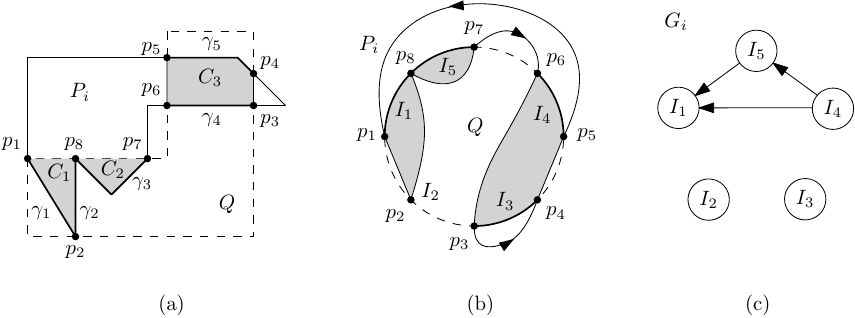}
  \caption{\small (a) $\myint{P_i}\cap \myint{Q}$ induces
    $R_i = \{C_1,C_2,C_3\}$, and each $\mybd{C_j} \cap \myint{Q}$
    induces simple paths $\{\gamma_1,\ldots,\gamma_5\}$.  (b) $p_8$
    serves as both an entry and exit point for the transition
    $C_1 \rightarrow C_2$.  The simple paths outside $\myint{Q}$, from
    $p_7$ to $p_6$ and from $p_5$ to $p_1$, represents the transitions
    $C_2\rightarrow C_3$ and $C_3\rightarrow C_1$, respectively.  (c)
    The directed graph $G_i$, constructed based on $P_i$ and $Q$ shown
    in (a).  }
  \label{fig:traversal.pieces}
\end{figure}

Finally, we relate these circular intervals to layer segments so as to
define a link instruction formally.  For any $j \in [t]$ and any
$k \in [\nlayers]$, consider the component $C_j \in R_i$ and the layer
$L_k$.  Each connected component of $\corridor \cap C_j$ (or
$L_k \cap C_j$) is bounded by four parts: continuous portions of the
inner and outer boundaries of $\corridor$ (or $L_k$), and two edges
from $\bdeset_{ij}$.  Each connected component of $L_k \cap C_j$ is
entirely contained within a unique connected component of
$\corridor \cap C_j$.  Moreover, the portion of $\mybd{Q}$ that bounds
each connected component of $\corridor \cap C_j$ corresponds to a
circular interval in $\mathcal{I}_{ij}$.  Thus, each layer segment of
$C_j$ in $L_k$ is uniquely associated with a circular interval in
$\mathcal{I}_{ij}$, in a one-to-one correspondence.  The number of
layer segments in $\corridor \cap C_j$ (or $L_k \cap C_j$) is
$|\mathcal{I}_{ij}|\cdot h$ (or $|\mathcal{I}_{ij}|$).

We revisit the link instruction that merges two layer segments
$Z_1\subseteq L_k \cap C_{j}$ and $Z_2\subseteq L_k \cap C_{j'}$ for
$j, j' \in [t]$ along layer $L_k$.  The link instruction reallocates
the region in $L_k$ spanned counterclockwise from $Z_1$ to $Z_2$.
Since $Z_1$ and $Z_2$ correspond to circular intervals in
$\mathcal{I}_{ij}$ and $\mathcal{I}_{ij'}$, respectively, each link
instruction can be represented by a triple $(I_1,I_2,k)$, where
$I_1,I_2 \in \mathcal{I}_{Q}$ and $k \in [\nlayers]$ (See
Figure~\ref{fig:reallocation}).  Note that
$(I_1,I_2,k)\neq (I_2,I_1,k)$, as the counterclockwise span from $Z_1$
to $Z_2$ differs from that in the reverse order.

\subsection{Graph for encoding link instructions}\label{subsubsec.graph.link.instruction}
Recall that $R_i = \{C_1,C_2,\ldots, C_t\}$ is the set of closures of
connected pieces of $\myint{P_i} \cap \myint{Q}$ such that each $C_j$
has a positive area and touches $\mybd{Q}$.  We construct a graph for
each $R_i$, which specifies how the elements of $R_i$ are to be
connected into a single piece within $Q$.

For each $j\in [t]$, $\mathcal{I}_{ij}$ consists of circular
intervals, each corresponding to a continuous path in
$\mybd{C_j} \cap \mybd{Q}$ along $\mybd{Q}$.  Let $\mathcal{I}_{ij}^+$
denote the subset of $\mathcal{I}_{ij}$ consisting of only those
intervals with positive length.  We define
$\mathcal{I}_{i}^+ = \bigcup_{j=1}^t \mathcal{I}_{ij}^+$ and
$\mathcal{I}^+_{Q} = \bigcup_{i=1}^m \mathcal{I}_{i}^+$.  Unlike
$\mathcal{I}_{Q}$, the set $\mathcal{I}^+_{Q}$ contains only
non-degenerate intervals, which are pairwise interior-disjoint.  Note
that $\{\mathcal{I}_{i}^+\}_{i=1,\ldots,m}$ forms a partition of
$\mybd{Q}$.

\subparagraph{Sequencing subpaths of $\mybd{P_i}$.}  For each
$C_j\in R_i$, $\mybd{C_j} \cap \myint{Q}$ consists of connected
components, each forming a continuous path connecting two points on
$\mybd{Q}$.  By the definition of $\mathcal{I}_{ij}$, the endpoints of
these paths correspond to the endpoints of the circular intervals in
$\mathcal{I}_{ij}$.  The total number of such paths in
$\mybd{C_j} \cap \myint{Q}$ is $|\mathcal{I}_{ij}|$. See
Figure~\ref{fig:traversal.pieces}(a).

We traverse $\mybd{P_i}$ counterclockwise, starting at any point on
$\mybd{C_1} \cap \myint{Q}$ and completing a full circuit.  During the
traversal, each path in $\mybd{C_j} \cap \myint{Q}$ is visited exactly
once for each $j \in [t]$, except for the path containing the starting
point.  The order in which the continuous paths are visited is denoted
by
$\gamma_1 \rightarrow \gamma_2 \rightarrow \cdots \rightarrow
\gamma_l$, where $\gamma_1 = \gamma_l$ is the path in
$\mybd{C_1}\cap \myint{Q}$ containing the starting point, and
$l = \sum_{j=1}^t |\mathcal{I}_{ij}| +1 = |\mathcal{I}_i|+1$.

Consider two consecutive paths $\gamma_k$ and $\gamma_{k+1}$ for any
$k \in [l-1]$.  Each path is derived from some piece in $R_i$; there
exist $j, j' \in [t]$ such that
$\gamma_k \subseteq \mybd{C_{j}} \cap \myint{Q}$ and
$\gamma_{k+1} \subseteq \mybd{C_{j'}} \cap \myint{Q}$.  As we traverse
from $\gamma_k$ to $\gamma_{k+1}$, we exit $\myint{Q}$ through an
endpoint of $\gamma_k$ and re-enter $\myint{Q}$ through an endpoint of
$\gamma_{k+1}$.

\begin{lemma}\label{lem:dis.interval.pos.length}
  Let $(\gamma_k, \gamma_{k+1})$ be a pair of consecutive paths for
  $k \in [l-1]$, where
  $\gamma_k \subseteq \mybd{C_{j}} \cap \myint{Q}$ and
  $\gamma_{k+1} \subseteq \mybd{C_{j'}} \cap \myint{Q}$ for some
  $j,j'\in[t]$.  The traversal from $\gamma_k$ to $\gamma_{k+1}$ exits
  and re-enters $\myint{Q}$ through intervals
  $I_1 \in \mathcal{I}_{ij}$ and $I_2 \in \mathcal{I}_{ij'}$,
  respectively.  If $j = j'$, then $I_1 = I_2$. Otherwise, $I_1$ and
  $I_2$ are interior-disjoint, non-degenerate intervals.
\end{lemma}
\begin{proof}
  The endpoint of $\gamma_k$ is an endpoint of an interval in $\mathcal{I}_{ij}$. 
  Since $C_{j}$ is a simple polygon in $Q$, the circular intervals in $\mathcal{I}_{ij}$ are pairwise disjoint. 
  Thus, there exists a unique interval $I_1 \in \mathcal{I}_{ij}$ 
  that contains the exit point.
  Similarly, 
  there exists a unique interval $I_2 \in \mathcal{I}_{ij'}$ that contains the entry point.
  
  Let $p$ and $q$ denote these exit and entry points on $\mybd{Q}$, respectively.
  Assume that $j = j'$, and that the traversal exits $\myint{Q}$ at $p$ and re-enters at $q$ without visiting any other pieces in $R_i$. 
  Since $\mybd{P_i}$ is traversed in counterclockwise direction, 
  the interior of $P_i$ lies to the left of an observer walking along $\mybd{P_i}$. 
  When $p=q$, the traversal touches $\mybd{Q}$ at $p$ and immediately re-enters its interior. 
  It follows that $I_1$ and $I_2$ are the same degenerate interval in $\mathcal{I}_{ij}$.
  
  Consider the case that $p \neq q$ and $j = j'$. 
  Let $I_{p,q}$ denote the portion of $\mybd{Q}$ from $p$ to $q$ in counterclockwise order.
  We construct two paths in $P_i$ from $p$ to $q$ in order to show that $I_1 = I_2$. 
  The first path is the portion of $\mybd{P_i}$ from $p$ to $q$ in counterclockwise order. 
  Since $j = j'$, this path lies in $P \setminus \myint{Q}$. 
  The second path connects $q$ to $p$ that lies in $\myint{C_{j}}$ except for $p$ and $q$.
  As $C_{j}$ is a simple polygon with a positive area,
  this path always exists. 

  By concatenating these two paths, 
  we obtain a simple closed curve in $P_i$ that encloses $I_{p,q}$.
  Note that $P_i$ is simply connected, which means 
  the region bounded by any simple closed curve in $P_i$ lies entirely within $P_i$.
  Thus, $I_{p,q}\subseteq C_{j}$, 
  and both $p$ and $q$ are contained in the same interval in $\mathcal{I}_{ij}$, i.e., $I_1 = I_2$. 

  We now turn to the case that $j \neq j'$.
  Suppose that $I_1$ is degenerate, which is a single point $p \in \mybd{Q}$. 
  Then $p$ must be a vertex of $C_{j}$, and 
  there are two edges $e_1'$ and $e_2'$ of $C_{j}$ that are incident to $p$.   
  Let $e_1$ and $e_2$ 
  be the corresponding edges in $P_i$ that contain $e_1'$ and $e_2'$, respectively. 
  Since $P_i$ is simple, only two edges of $P_i$ are incident to $p$, namely $e_1$ and $e_2$. 
  Thus, traversal from $\gamma_k$ to $\gamma_{k+1}$ remains on the boundary of a single piece $C_{j}$, a contradiction.
  It follows that $I_1$ is non-degenerate, and similarly, so is $I_2$.
  As $\{\mathcal{I}_{1}^+ ,\ldots \mathcal{I}_{m}^+\}$ forms a partition of $\mybd{Q}$, 
  $I_1$ and $I_2$ are interior-disjoint, non-degenerate intervals.
\end{proof}

By Lemma~\ref{lem:dis.interval.pos.length}, a pair of distinct
intervals in $\mathcal{I}_{i}^+$ provides an exit-entry pair for each
transition between distinct pieces in $R_i$ in the counterclockwise
traversal of $\mybd{P_i}$.

\subparagraph{Construction of $G_i$.}  Based on the counterclockwise
traversal of $\mybd{P_i}$, we construct a directed graph
$G_i = (V_i,E_i)$ as follows.  Each vertex $v\in V_i$ corresponds to
an interval in $\mathcal{I}_{i}^+$, so $|V_i| = |\mathcal{I}_i^+|$.
Let $\gamma_1$ and $\gamma_2$ be consecutive paths in the traversal
where $\gamma_1 \subseteq \mybd{C_{j}}\cap \myint{Q}$ and
$\gamma_2 \subseteq\mybd{C_{j'}}\cap \myint{Q}$ with $j \neq j'$.  The
traversal from $\gamma_1$ to $\gamma_2$ encounters intervals $I_1$ and
$I_2$ in $\mathcal{I}^+_i$ that contain the exit and entry points,
respectively.  By Lemma~\ref{lem:dis.interval.pos.length}, these
intervals $I_1$ and $I_2$ are uniquely determined.  We add a directed
edge $e$ between vertices $v_1$ and $v_2$ corresponding to $I_1$ and
$I_2$, respectively.

Let $p$ and $q$ denote the exit and entry points of the traversal,
respectively; $p$ is the endpoint of $I_1$ and $q$ is the endpoint of
$I_2$.  The counterclockwise traversal on $\mybd{P_i}$ between
$\gamma_1$ and $\gamma_2$ follows a simple path along $\mybd{P_i}$
outside $Q$ that starts from $p$ and ends at $q$.  Let
$\gamma_{pq} \subseteq \mybd{P_i} \setminus \myint{Q}$ denote this
path from $p$ to $q$.  Observe that the exit and entry points may
coincide, i.e., $p = q$, if $I_1$ and $I_2$ are adjacent at $p$ and
$I_1$ lies counterclockwise from $I_2$ along $\mybd{Q}$.  In this
case, the direction of $e$ is assigned from $v_2$ to $v_1$.

Given that $p\neq q$, the simple path $\gamma_{pq}$ can be viewed as a
simple path connecting two distinct points on the boundary of the
circle and lying outside the circle.  The path $\gamma_{pq}$ can be
classified into one of two types, depending on how it winds around the
circle.  Formally, $\gamma_{pq}$ is homotopic to a directed path in
$\mathbb{R}^2 \setminus \myint{Q}$ that winds around $\mybd{Q}$ in
either counterclockwise or clockwise direction.  Note that it cannot
wind around the boundary more than once since $\gamma_{pq}$ is simple.
We assign the direction of $e$ from $v_1$ to $v_2$ if the path is of
the counterclockwise type, and from $v_2$ to $v_1$, otherwise.
Figure~\ref{fig:traversal.pieces}(b) illustrates $\gamma_{pq}$ for
both cases where $p=q$ and $p\neq q$.  The same rule is applied to
assign directions to all other edges in $G_i$.

In summary, we construct the directed graph $G_i$ for $i\in[m]$ to
represent how the disjoint components of $R_i$ are to be connected
into a single piece within $Q$.  Each vertex corresponds to an
interval in $\mathcal{I}_{i}^+$.  Each directed edge $e = (v_1,v_2)$
of $G_i$ represents a link instruction $(I_1, I_2, k)$, but $k$ is not
specified yet.  See Figure~\ref{fig:traversal.pieces}(c) for an
illustration of the directed graph $G_i$ that defines three link
instructions, where $P_i$ and $Q$ are as shown in
Figure~\ref{fig:traversal.pieces}(a).

\subsection{Layer assignments for link instructions}\label{subsubsec.layer.assignment}
For $i\in [m]$, let $G_i$ be the directed graph for
$R_i = \{C_1,C_2,\ldots ,C_t\}$.  Each edge in $G_i$ represents a link
instruction that reallocates layer segments within a specific layer.
The goal is to assign layers to link instructions so that no redundant
reallocation of layer segments is allowed and all elements in $R_i$
are eventually merged into a single piece for every $i =1,\ldots,m$.

The underlying graph of a digraph is its undirected version, obtained
by ignoring the directions of all edges.  Consider a connected
component $T$ of the underlying graph of $G_i$.  We assume that $T$
contains more than one vertex, as a component of size 1 does not
indicate any link instruction.  Each vertex of $T$ corresponds to a
circular interval in $\mathcal{I}^+_i$.  Let $\mathcal{I}_T$ be the
subset of $\mathcal{I}^+_i$ such that
$\mathcal{I}_T = \{I \in \mathcal{I}^+_i \mid I \text{ corresponds to
  a vertex in } T\}$.

\begin{figure}[!b]
  \centering
  \includegraphics[width=0.65\textwidth]{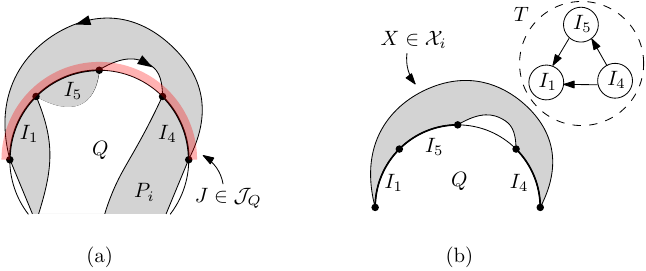}
  \caption{\small (a) The combine step merges $I_1$, $I_5$, and $I_4$
    to a circular interval $J\in \mathcal{J}_Q$.  (b) The interval $J$
    is derived from a connected component $T$ of the underlying graph
    of $G_i$.  This corresponds to a connected region
    $X \in \mathcal{X}_i$ that is bounded by the circular intervals
    associated with $T$.  }
  \label{fig:path.homotopy}
\end{figure}

We iteratively perform a combine step to merge all circular intervals
in $\mathcal{I}_T$ into a single circular interval on $\mybd{Q}$.  The
process begins with an initial set $S = \{u\}$ and a circular interval
$J_S$, where $u$ is any vertex of $T$ and $J_S$ is the circular
interval in $\mathcal{I}_T$ corresponding to $u$.  At each step, any
vertex $v \in T \setminus S$ is chosen if there exists an edge between
$v$ and some vertex $w \in S$.  Let $I_v \in \mathcal{I}_T$ be the
circular interval corresponding to $v$.  We add $v$ to $S$ and update
$J_S$ to be the smallest circular interval on $\mybd{Q}$ that spans
from $I_v$ to $J_S$ in counterclockwise direction if the edge is
directed from $v$ to $w$, or in clockwise direction otherwise.  Once
$S = T$, the process returns a single circular interval on $\mybd{Q}$,
denoted by $J_T$, corresponding to the connected component $T$ of
$G_i$.  See Figure~\ref{fig:path.homotopy}(a) for an illustration of
this process.

\[\mathcal{J}_Q\coloneqq\{J_T \mid T\text{ is a connected
    component of }G_i \text{ with } |T|>1, \text{ for } i=1,2,\ldots,
  m\}.\] 
  Throughout, we use $I$ to denote intervals in
$\mathcal{I}_{Q}$ and $J$ for those in $\mathcal{J}_Q$ to emphasize
their distinct roles.  Each $J \in \mathcal{J}_Q$ is non-degenerate.
For any $J_1,J_2 \in \mathcal{J}_Q$ with
$\myint{J_1} \cap \myint{J_2} \neq \emptyset$, there exists a proper
containment between $J_1$ and $J_2$.  To prove this, we present a
lemma relating connectivity in the underlying graph of $G_i$ to the
existence of a path contained in a connected component of
$P_i \setminus \myint{Q}$, as illustrated in
Figure~\ref{fig:path.homotopy}(b).

\begin{lemma}\label{lem:path.connected}
  The vertices corresponding to $I_1$ and $I_2$ in $\mathcal{I}^+_i$
  are connected in the underlying graph of $G_i$ if and only if there
  is a simple path connecting $p\in I_1$ to $q\in I_2$ contained in
  $P_i \setminus \myint{Q}$.
\end{lemma}
\begin{proof}
  Let $v_1$ and $v_2$ be vertices of $G_i$ that correspond to $I_1$ and $I_2$, respectively.
  Assume that $v_1$ and $v_2$ are adjacent in $G_i$. 
  Then, $I_1$ and $I_2$ provide an exit-entry pair 
  for a transition between distinct pieces in $R_i$ while traversing $\mybd{P_i}$.
  There is a path $\gamma$ along $\mybd{P_i}$
  that starts from the endpoint of $I_1$ and terminates at the endpoint of $I_2$.
  We add the portions of $I_1$ and $I_2$ to $\gamma$ so that 
  the extended path connects $p$ and $q$. 
  Since $\gamma$ does not enter the interior of $Q$, 
  the extended path lies in $P_i \setminus \myint{Q}$.

  Consider the case that $v_1$ and $v_2$ are not adjacent, 
  but are connected in the underlying graph of $G_i$.
  In that case, 
  there exists an indirect path between $v_1$ and $v_2$ in the underlying graph: 
  $v_1= u_1 - u_2 - \cdots - u_{r+1}= v_2 $ 
  for $\{u_1,\ldots, u_{r+1}\} \subseteq V_i$. 
  For each consecutive pair $(u_k, u_{k+1})$, 
  we obtain a simple path $\gamma'_k \subseteq P_i \setminus \myint{Q}$  
  that connects two points on their corresponding intervals in $\mathcal{I}_i^+$. 
  By choosing points appropriately, 
  we have simple paths $\{\gamma'_1, \ldots, \gamma'_r\}$ such that
  $\gamma'_1$ starts at $p$, $\gamma'_r$ terminates at $q$, and 
  each consecutive pair intersects at a common point.
  By concatenating these paths, 
  we obtain a new path in $P_i \setminus \myint{Q}$ that connects $p$ to $q$. 

  Now assume that there is a simple path $\gamma \subseteq P_i \setminus \myint{Q}$ 
  that connects $p \in I_1$ to $q \in I_2$. 
  Since $\gamma$ is connected, 
  there exists a connected component $X$ of $P_i \setminus \myint{Q}$ that contains $\gamma$. 
  The boundary of $X$
  consists of circular intervals in $\mathcal{I}^+_i$ and connected subpaths of $\mybd{P_i}$. 
  
  Let $\mathcal{I}_X \subseteq \mathcal{I}_i^+$ 
  be the set of circular intervals in $\mybd{X} \cap \mybd{Q}$.   
  Observe that $I_1$ and $I_2$ are contained in $\mathcal{I}_X$.
  When traversing $\mybd{X}$ in counterclockwise order,  
  every interval in $\mathcal{I}_X$ is encountered at least once.  
  Moreover, every transition between distinct intervals in $\mathcal{I}_X$ requires 
  traversing a portion of $\mybd{P_i}$ that connects them and lies outside $\myint{Q}$.
  Let $I_1 \rightarrow \cdots \rightarrow I_r$ be the sequence of intervals in $\mathcal{I}_X$
  in the order they are visited during the traversal. 
  We claim that $I_k$ and $I_{k+1}$ are derived from distinct elements in $R_i$ for every $k =1,\ldots, r-1$.
  
  Suppose that 
  both $I_k$ and $I_{k+1}$ are contained in $\mathcal{I}^+_{ij}$ for some $j\in [t]$.
  They are not adjacent along $\mybd{Q}$, 
  as $C_j$ is a simple polygon and its intersection with $\mybd{Q}$ is composed of pairwise disjoint intervals. 
  We have a simple path $\gamma^{\mathrm{out}} \subseteq \mybd{X}$ 
  that connects the endpoints of $I_k$ and $I_{k+1}$ and lies outside $\myint{Q}$. 
  Moreover, we have another simple path $\gamma^{\mathrm{in}} \subseteq C_j$ that connects the same endpoints.
  Then, we construct a loop by concatenating $\gamma^{\mathrm{out}}$ and $\gamma^{\mathrm{in}}$ 
  at their shared endpoints. 
  
  Except at their shared endpoints, the two paths lie in $P_i \setminus Q$ and $P_i \cap \myint{Q}$, respectively.
  Since $P_i$ is simply connected and the loop lies in $P_i$, 
  the interior region of the loop is also contained in $P_i$. 
  The portion of $\mybd{Q}$ between $I_k$ and $I_{k+1}$ is contained in the loop, 
  and hence the portion must be contained in $\mybd{C_j}$. 
  This contradicts that 
  $I_k$ and $I_{k+1}$ are disjoint intervals in $\mathcal{I}_{ij}^+$.
  Thus, each transition of intervals in $\mathcal{I}_X$ 
  induces an edge in $G_i$ between corresponding vertices, 
  and thus their corresponding vertices are connected in the underlying graph of $G_i$.
\end{proof}

Let $\mathcal{X}_i$ be the set of connected components in
$P_i \setminus \myint{Q}$.  Then each element in $\mathcal{X}_i$ is a
closed connected set.  Moreover, it is a (weakly) simple polygon.

For $p, q \in \mybd{Q}$, let $\paths(p, q)$ denote the set of all
simple paths from $p$ to $q$ contained in $P_i \setminus \myint{Q}$.
For $I_1,I_2 \in \mathcal{I}^+_i$, we define
$\paths(I_1,I_2) = \bigcup_{p \in I_1, q\in I_2} \paths(p,q)$.
Lemma~\ref{lem:path.connected} guarantees the existence of a path
between $I_1$ and $I_2$ whenever their corresponding vertices are
connected in the underlying graph of $G_i$.  Furthermore, all such
paths share a consistent topological behavior, such as winding around
$\mybd{Q}$ in the same direction.

\begin{corollary}\label{cor:unique.path.structure}
  Let $I_1,I_2 \in \mathcal{I}^+_i$ with $I_1 \neq I_2$. If
    $\paths(I_1,I_2) \neq \emptyset$, then all paths in
    $\paths(I_1,I_2)$ wind around $\mybd{Q}$ in the same direction,
  and there is a unique component $X \in \mathcal{X}_i$ that
    contains all such paths.
\end{corollary}
\begin{proof}
  Assume that $\paths(I_1,I_2) \neq \emptyset$, since the statement is trivially true otherwise.
  Let $p_1, p_2 \in I_1$ and $q_1, q_2 \in I_2$ be points on $\mybd{Q}$. 
  By Lemma~\ref{lem:path.connected}, there exists $\gamma_1 \in \paths(I_1,I_2)$ and $\gamma_2\in \paths(I_1,I_2)$ that 
  connects $p_1$ to $q_1$ and $p_2$ to $q_2$, respectively.
  Moreover, $\gamma_1 \subseteq X_1$ and $\gamma_2 \subseteq X_2$ for some $X_1,X_2\in \mathcal{X}_i$.
  
  We obtain a closed loop $\ell$ in $P_i\setminus \myint{Q}$ 
  by connecting the endpoints of $\gamma_1$ and $\gamma_2$ through subpaths of $I_1$ and $I_2$. 
  Since $\ell$ intersects both $X_1$ and $X_2$, and each $X \in \mathcal{X}_i$ is a connected component, 
  we have $X_1=X_2$.

  Now suppose that $\gamma_1$ and $\gamma_2$ wind around $\mybd{Q}$ in different directions.
  That is, $\gamma_1$ is homotopic to a path along $\mybd{Q}$ 
  that winds around $\mybd{Q}$ in counterclockwise direction, 
  while $\gamma_2$ is homotopic to one that winds around it in clockwise direction.
  Then the loop $\ell$ encircles $Q$ exactly once.
  Since $P_i$ is simply connected, this loop must be null-homotopic in $P_i$.
  This implies $Q \subseteq P_i$, and then, $I_1 =I_2 = \mybd{Q}$. 
  It is a contradiction since the set $\mathcal{I}_{Q}$ does not contain a loop.
\end{proof}

Using Lemma~\ref{lem:path.connected} and
Corollary~\ref{cor:unique.path.structure}, we prove a proper
containment relation.

\begin{lemma}\label{lem:proper.containment.intervals}
  For any distinct $J_1,J_2\in\mathcal{J}_Q$ with
  $\myint{J_1} \cap \myint{J_2} \neq \emptyset$, $J_1 \subsetneq J_2$
  or $J_2 \subsetneq J_1$.
\end{lemma}
\begin{figure}[!ht]
  \centering
  \includegraphics[width=.9\textwidth]{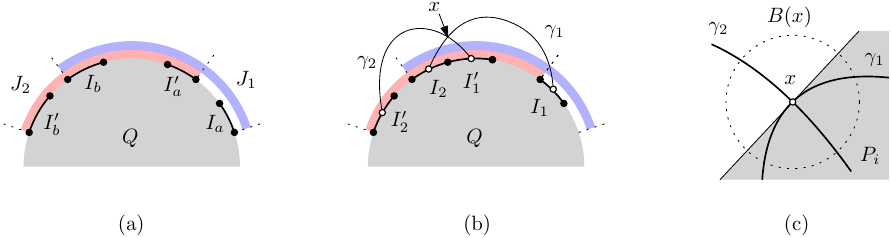}
  \caption{\small 
  (a) $J_1$ (or $J_2$) is the smallest circular interval that spans $I_a$ to $I_b$ (or $I_a'$ to $I_b'$) in counterclockwise direction. 
  (b) The paths $\gamma_1$ and $\gamma_2$, which lie outside $\myint{Q}$, connect $\myint{I_1}$ to $\myint{I_2}$, and 
  $\myint{I_1'}$ to $\myint{I_2'}$, respectively. They intersect at a common interior point $x$. 
  (c) When $i \neq i'$, $\gamma_2$ enters the interior of $P_i$ near the crossing point $x$. 
  }
  \label{fig:assign.layers}
\end{figure}
\begin{proof}
  Each $J \in \mathcal{J}_Q$ is constructed by merging intervals from $\mathcal{I}_i^+$
  that correspond to 
  a connected component $T$ in $G_i$, for some $i \in [m]$.
  Let $\mathcal{I}(J)$ be the set of intervals in $\mathcal{I}^+_{Q}$ corresponding to the vertices of $T$. 
  Then, $J$ is the smallest circular interval that spans $I_1$ to $I_2$ 
  in counterclockwise direction for some $I_1,I_2 \in \mathcal{I}(J)$. 
  
  Let $I_a,I_b \in \mathcal{I}(J_1)$ and $I'_a, I'_b \in \mathcal{I}(J_2)$ 
  such that $J_1$ is the smallest interval spanning $I_a$ to $I_b$ in counterclockwise direction 
  and $J_2$ is defined analogously as the interval from $I_a'$ to $I_b'$ in the same direction. 
  See Figure~\ref{fig:assign.layers}(a).
  Observe that all intervals in $\mathcal{I}(J_1)$ 
  are interior-disjoint from those in $\mathcal{I}(J_2)$ if $J_1 \neq J_2$. 
  It follows that $I_a$, $I_b$, $I_a'$, and $I_b'$ are all distinct circular intervals.

  Suppose that $J_1$ and $J_2$ properly intersect each other, that is 
  $J_1 \not\subseteq J_2$, $J_2 \not\subseteq J_1$, 
  and $\myint{J_1} \cap \myint{J_2} \neq \emptyset$.
  We only consider the case that $I_a' \subseteq J_1$ and $I_b' \not\subseteq J_1$, 
  since the other can be handled analogously.

  Since $I_a' \subseteq J_1$ and $I_b' \not\subseteq J_2$, 
  the intervals $I_a$, $I_a'$, $I_b$, and $I_b'$ appear in counterclockwise order along $\mybd{Q}$. 
  Let $\mathcal{I}(J_1)[I_a',I_b]$ denote the subset of $\mathcal{I}(J_1)$ 
  that appear along $\mybd{Q}$ from $I_a'$ to $I_b$ in counterclockwise order, including $I_b$.
  Let $\mathcal{I}(J_1)[I_a,I_a']$ denote the subset of $\mathcal{I}(J_1)$ consisting of intervals that do not belong to $\mathcal{I}(J_1)[I_a',I_b]$, 
  i.e., $\mathcal{I}(J_1) = \mathcal{I}(J_1)[I_a,I_a'] \cup \mathcal{I}(J_1)[I_a',I_b]$.
  Note that $\mathcal{I}(J_1)[I_a,I_a']$ and $\mathcal{I}(J_1)[I_a',I_b]$ are non-empty.
  Then, there exist $I_1 \in \mathcal{I}(J_1)[I_a,I_a']$ and $I_2 \in \mathcal{I}(J_1)[I_a',I_b]$ such that 
  the vertices corresponding to $I_1$ and $I_2$ are adjacent in $G_i$. 
  If no such $I_1$ and $I_2$ exist, intervals of $\mathcal{I}(J_1)[I_a,I_a']$ and $\mathcal{I}(J_1)[I_a',I_b]$ are not merged into $J_1$ 
  in the process of combine steps. 

  Similarly, let $\mathcal{I}(J_2)[I_a',I_2]$ denote the subset of $\mathcal{I}(J_2)$
  whose elements appear along $\mybd{Q}$ from $I_a'$ to $I_2$ in counterclockwise order, including $I_a'$,
  and let $\mathcal{I}(J_2)[I_2, I_b'] = \mathcal{I}(J_2)\setminus \mathcal{I}(J_2)[I_a',I_2]$. 
  By the same reasoning, there exist $I_1' \in \mathcal{I}(J_2)[I_a',I_2]$ and $I_2' \in \mathcal{I}(J_2)[I_2, I_b']$ such that 
  the corresponding vertices are adjacent in $G_{i'}$ for some $i' \in [m]$.
  Observe that $I_1, I_1', I_2$, and $I_2'$ appear in counterclockwise order along $\mybd{Q}$. 

  As $I_1$ and $I_2$ are non-degenerate intervals, 
  we choose points $p \in \myint{I_1}$ and $q\in \myint{I_2}$.
  By Lemma~\ref{lem:path.connected} and Corollary~\ref{cor:unique.path.structure}, 
  there exists a simple path from $p$ to $q$ in $P_i \setminus \myint{Q}$.
  Moreover, this path is continuously deformed into one that winds around $\mybd{Q}$ counterclockwise.
  The same argument holds for two intervals $I_1'$ and $I_2'$ in $\mathcal{I}(J_2)$.

  For each pair $(I_1,I_2)$ and $(I_1',I_2')$, 
  we have a simple path that connects their interior points and lies outside $\myint{Q}$. 
  Let $\gamma_1$ and $\gamma_2$ denote these paths for $(I_1,I_2)$ and $(I_1',I_2')$, respectively.
  Then, $\gamma_1$ is the simple path that starts from a point in $\myint{I_1}$ and terminates at $\myint{I_2}$.
  The deformed path of $\gamma_1$ 
  on $\mybd{Q}$ must include $I_a'$. 
  Similarly, 
  the deformation of $\gamma_2$ on $\mybd{Q}$ includes $I_2$. 
  Then, $\gamma_1 \subseteq P_i$ and $\gamma_2 \subseteq P_{i'}$ must cross at their interior point, 
  denoted by $x$. See Figure~\ref{fig:assign.layers}(b).
  
  If $i \neq i'$, $x$ lies on $\mybd{P_i} \cap \mybd{P_{i'}}$
  since $\myint{P_i} \cap \myint{P_{i'}} = \emptyset$.
  Considering a sufficiently small ball $B(x)$ centered at $x$,  
  $\gamma_1$ divides the ball into two connected regions, 
  where one region is entirely contained in $P_i$. 
  Since $\gamma_2$ crosses $\gamma_1$ at $x$, it must intersect 
  the interior of the region of $P_i$, 
  a contradiction. Figure~\ref{fig:assign.layers}(c) illustrates
  this configuration.

  If $i = i'$, 
  the intersection between $\gamma_1$ and $\gamma_2$ 
  allows us to construct a new simple path in $P_i \setminus \myint{Q}$ 
  that connects a point in $I_1$ to a point in $I_1'$.
  By Lemma~\ref{lem:path.connected}, 
  the vertices corresponding to $I_1$ and $I_1'$ are connected in the underlying graph of $G_i$. 
  Then, $I_1$ and $I_2$ are combined into a single interval in $\mathcal{J}_Q$, a contradiction. 
  Whether $i = i'$ or not,
  assuming $J_1 \not\subseteq J_2$ and $J_2 \not\subseteq J_1$ leads to a contradiction.
\end{proof}

By Lemma~\ref{lem:proper.containment.intervals}, we construct the
circular interval graph $G_{\mathcal{J}}$, where each node corresponds
to a circular interval in $\mathcal{J}_Q$, and a directed edge from
$J_1$ to $J_2$ is added if $J_1 \subsetneq J_2$ for
$J_1, J_2 \in \mathcal{J}_Q$.  Then, $G_{\mathcal{J}}$ is acyclic as
each edge corresponds to a strict containment.

\subparagraph{Transitive reduction of $G_{\mathcal{J}}$.}  The graph
$G_{\mathcal{J}}$ represents the transitive closure of the proper
containment relation.  That is, for $J_1, J_2 \in \mathcal{J}_Q$,
there is a directed edge from $J_1$ to $J_2$ in $G_{\mathcal{J}}$ if
there exists some $J' \in \mathcal{J}_Q$ such that $(J_1, J')$ and
$(J' , J_2)$ are directed edges in $G_{\mathcal{J}}$.  A
\emph{transitive reduction} of $G_{\mathcal{J}}$ is a directed graph
on the same vertex set with the minimum number of edges that preserves
all reachability relations of $G_{\mathcal{J}}$.  Since the transitive
reduction of a DAG is unique~\cite{BangJensen2008}, we denote the
transitive reduction of $G_{\mathcal{J}}$ by $\trg{G}_{\mathcal{J}}$.
Let $U\trg{G}_{\mathcal{J}}$ denote the underlying graph of
$\trg{G}_{\mathcal{J}}$.

\begin{lemma}\label{lem:inbranching}
  Any vertex of $\trg{G}_{\mathcal{J}}$ has out-degree at most one,
  and $U\trg{G}_{\mathcal{J}}$ is acyclic.
\end{lemma}
\begin{proof}
  Suppose, for the sake of contradiction, that there is $J_1 \in \mathcal{J}_Q$ with an out-degree at least two.
  There exist two distinct intervals
  $J_2, J_3 \in \mathcal{J}_Q$ such that $(J_1, J_2)$ and $(J_1,J_3)$ are directed edges in $\trg{G}_{\mathcal{J}}$. 
  Since both $J_2$ and $J_3$ properly contain $J_1$,  
  their interiors must intersect over the interior of $J_1$.  
  By Lemma~\ref{lem:proper.containment.intervals}, either $J_2 \subsetneq J_3$ or $J_3 \subsetneq J_2$. 
  Thus, there is a directed edge between $J_2$ and $J_3$ in $G_{\mathcal{J}}$. 
  Without loss of generality, assume that 
  the edge is directed from $J_2$ to $J_3$.

  We construct a graph $G'_{\mathcal{J}}$ from $\trg{G}_{\mathcal{J}}$
  by deleting the edge $(J_1,J_3)$ and 
  adding the edge $(J_2, J_3)$ 
  if it is not already present in $\trg{G}_{\mathcal{J}}$. 
  Note that $G'_{\mathcal{J}}$ is a subgraph of $G_{\mathcal{J}}$ and the number of edges in $G'_{\mathcal{J}}$ 
  is no greater than that of $\trg{G}_{\mathcal{J}}$. 
  Moreover, $G'_{\mathcal{J}}$ preserves the same reachability relations as $\trg{G}_{\mathcal{J}}$, 
  since the only modification from $\trg{G}_{\mathcal{J}}$ to $G'_{\mathcal{J}}$ is replacing the direct path from $J_1$ to $J_3$ 
  by the indirect path $J_1\rightarrow J_2 \rightarrow J_3$. 
  Therefore, $G'_{\mathcal{J}}$ is also a transitive reduction of $G_{\mathcal{J}}$, 
  contradicting that the transitive reduction of a directed acyclic graph is unique.

  Now suppose that there is a cycle $T$ in $U\trg{G}_{\mathcal{J}}$.
  Let $\trg{G}_{\mathcal{J}}[T]$ denote the subgraph of $\trg{G}_{\mathcal{J}}$ induced by $T$.
  As the sum of all out-degrees in a graph equals to the number of edges, 
  the sum of out-degrees in $\trg{G}_{\mathcal{J}}[T]$ is $|T|$. 
  By the earlier argument, every vertex of $T$ must have an out-degree exactly one.
  Then, $T$ forms a directed cycle in $\trg{G}_{\mathcal{J}}$, 
  contradicting that $\trg{G}_{\mathcal{J}}$ is a directed acyclic graph.
  Thus, $U\trg{G}_{\mathcal{J}}$ is acyclic, and it forms a forest.
\end{proof}

By Lemma~\ref{lem:inbranching}, $U\trg{G}_{\mathcal{J}}$ is a forest.
Each directed tree $T$ of $\trg{G}_{\mathcal{J}}$ has total
out-degrees $|T|-1$, implying that exactly one vertex has out-degree
zero and all others have out-degree one. This structure corresponds to
an \emph{in-branching} tree~\cite{BangJensen2008}.

\subparagraph{Layer assignment of link instructions.}  In each
in-branching tree of $\trg{G}_{\mathcal{J}}$, the unique vertex with
out-degree zero is called the root.  The level of a vertex is defined
as the number of edges on the path from the root to that vertex plus
one, where the level of the root is one.  The level of each interval
$J \in \mathcal{J}_Q$ is defined to be the level of the corresponding
vertex in $\trg{G}_{\mathcal{J}}$.

Recall that each circular interval $J \in \mathcal{J}_Q$ is formed by
merging intervals in $\mathcal{I}_i^+$ for some $i \in [m]$.  These
intervals correspond to the vertices of a connected component $T$ of
$G_i$, where each edge of $T$ encodes a link instruction that merges
two elements in $R_i$.  Let $\inst(J)$ denote the set of link
instructions corresponding to the edges of $T$.  We assign the
  $k$-th layer to every link instruction in $\inst(J)$ if $J$ is at
level $k$ in $\trg{G}_{\mathcal{J}}$.  Note that the maximum level in
$\trg{G}_{\mathcal{J}}$ is at most
$\lfloor|\mathcal{I}^+_{Q}|/2\rfloor$ if each connected component of
$G_i$ has size two for all $i = 1,2,\ldots, m$.  Thus, we set the
number of layers in the corridor $\corridor$ to this maximum level:
$h = \lfloor|\mathcal{I}^+_{Q}|/2\rfloor$.

For each $J\in \mathcal{J}_Q$, all link instructions in $\inst(J)$
reallocate layer segments within a single layer $L_k$ to a common
piece $P_i$, where $i \in [m]$ and $k \in [\nlayers]$ are determined
by $J$.  Let $I_a, I_b \in \mathcal{I}_{i}^+$ such that $J$ is the
smallest circular interval that spans from $I_a$ to $I_b$ in
counterclockwise order.  Let $Z_a, Z_b\subseteq L_k$ be the layer
segments corresponding to $I_a$ and $I_b$, respectively.  Applying
$\inst(J)$ is equivalent to reallocating every $Z\in L_k[Z_a,Z_b]$ to
$P_i$.

In summary, $\inst_Q = \bigcup_{J\in \mathcal{J}_Q} \inst(J)$ defines
the set of all link instructions for reconfiguration.  Applying
$\inst_Q$ to $\Pi[Q]$ yields a partition
$Q = Q_1^\ast \cup \cdots \cup Q_m^\ast$, where each $Q_i^\ast$ is a
region assigned to $P_i$ within $Q$.  When $Q$ is
$\overline{W}$-convex with respect to $P$, each $Q_i^\ast$ is
connected and satisfies both constraints $W$ and $U$.

\section{Analysis of reconfigured non-guillotine partitions}\label{sec.analysis.reconfiguration}
Assuming that the link instructions in $\inst_Q$ have been
  applied to $\Pi[Q]$ in an arbitrary order, we verify the following
  statements in order to prove Theorem~\ref{thm:monotonicity.containment}.
\begin{enumerate}[{(1)}]\denseitems
\item For each $i = 1, \ldots, m$, the region allocated to $P_i$
  within $Q$ forms a single connected piece.
\item The reconfiguration of $\Pi[Q]$ is a solution to the problem
  $\prob{Q,W,U}$ if $Q$ is $\overline{W}$-convex with respect to $P$.
\end{enumerate}
It follows from statements (1) and (2) that the reconfigured
partition of $Q$ satisfies the constraints $W$ and $U$ while ensuring
that the number of pieces does not exceed $m$.

\subsection{Connectivity in reconfigured partitions}
For $i\in [m]$, let $\inst_i$ denote the subset of $\inst_Q$
consisting of instructions of the form $(I_1, I_2, k)$, where
$I_1, I_2 \in \mathcal{I}^+_i$ and $k \in [\nlayers]$.  Then,
$\inst_Q = \bigcup_{i\in[m]}\inst_i$.  We first show that the link
instructions in $\inst_i$ merge all elements of $R_i$ into a single
connected piece.  We then verify that applying the link
instructions in $\inst_Q \setminus \inst_i$ does not disconnect the
merged piece. Furthermore, the resulting reconfigured
  partition of $Q$ is invariant under the order in which the link
  instructions in $\inst_Q$ are applied to $\Pi[Q]$.

\subparagraph{Merging $R_i$ via $\inst_i$.}  We show that any two
pieces $C_{j}, C_{j'}\in R_i$ are connected by some link instructions
in $\inst_i$.  During the traversal on $\mybd{P_i}$, each continuous
path of $\mybd{C_j} \cap \myint{Q}$ is encountered at least once for
every $j \in [t]$, and we denote the sequence of these paths in the
order they are visited as $(\gamma_1, \gamma_2, \ldots, \gamma_l)$,
where $\gamma_1 = \gamma_l$.  Then, there exist indices
$k, k' \in [l]$ such that
$\gamma_{k} \subseteq \mybd{C_{j}} \cap \myint{Q}$ and
$\gamma_{k'} \subseteq \mybd{C_{j'}} \cap \myint{Q}$.  Since it is a
cyclic sequence with $\gamma_1 = \gamma_l$, we assume without loss of
generality that $k < k'$.

Recall that an edge of $G_i$ is added whenever two consecutive paths
in the sequence are derived from distinct elements in $R_i$.  This
indicates that a transition between them occurs during the traversal.
The link instruction associated with this edge merges the
corresponding elements in $R_i$.  The subsequence
$(\gamma_{k}, \gamma_{k+1}, \ldots, \gamma_{k'})$ contains multiple
transitions between distinct elements in $R_i$, starting from $C_{j}$
and eventually reaching $C_{j'}$.  Applying all link instructions
associated with the transitions in $(\gamma_{k}, \ldots, \gamma_{k'})$
results in $C_{j}$ and $C_{j'}$ being merged into a single connected
piece.

Let $Q_i$ denote the subregion of $Q$ resulting from merging
elements of $R_i$ via $\inst_i$, with no other instructions
  applied. Then, $Q_i$ consists of all elements in $R_i$ and layer
segments reallocated from other pieces:
$Q_i = \Bigl(\bigcup_{j \in [t]} C_j\Bigr) \cup \segout_i$, where
$\segout_i$ is the set of layer segments reallocated to $P_i$ by
$\inst_i$.

The region $\mathbf{L} = L_1 \cup \cdots\cup L_{\nlayers}$ denotes the
union of all layers in $Q$.  For each $C_j \in R_i$,
$\mathbf{L} \cap C_j$ consists of layer segments within $C_j$.  By
construction of layers, $C_j \setminus \mathbf{L}$ is connected and
non-empty.  We refer to this region as a \emph{core} of $C_j$, denoted
by $\core{C_j} = C_j \setminus \mathbf{L}$.  Each $C_j$ consists of
its core together with the layer segments it contains.  Let $\segin_i$
denote the set of all layer segments in $\mathbf{L} \cap C_j$ over all
$C_j \in R_i$.  Therefore, for $i\in [m]$, we have
$Q_i = \Bigl(\bigcup_{j \in [t]} \core{C_j}\Bigr) \cup \segin_i \cup
\segout_i.$
Figure~\ref{fig:applying.instructions}(a) illustrates the parts of
$Q_i$: cores and layer segments.  The shapes are drawn schematically
to reflect the topological structure, rather than an exact polygonal
description.

\begin{figure}[!t]
  \centering
  \includegraphics[width=0.85\textwidth]{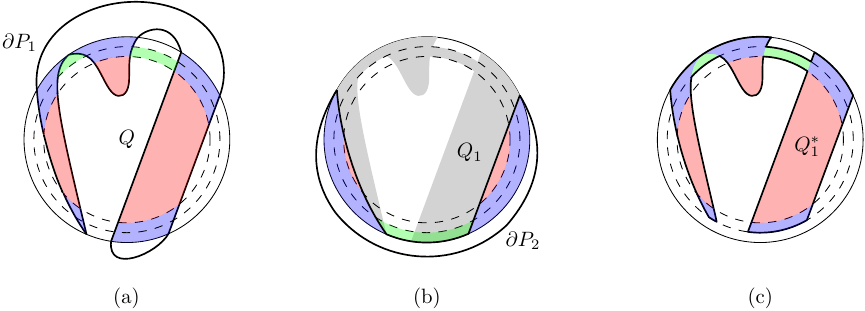}
  \caption{\small The corridor of $Q$ has two layers, with $|R_1| = 3$
    and $|R_2| = 2$.  (a) Red regions represent $\{\core{C} \mid$ for
    $C \in R_1$\}, while blue and green regions indicate layer
    segments in $\segin_1$ and $\segout_1$, respectively.  (b) $Q_1$
    is formed by merging $R_1$ according to $\inst_{1}$, and some
    layer segments in $\segin_1$ are reallocated by $\inst_2$.  (c)
    $Q_1^\ast$ is obtained by applying $\inst_Q\setminus \inst_1$ to
    $Q_1$.  }
  \label{fig:applying.instructions}
\end{figure}

We now turn to the link instructions in $\inst_Q \setminus \inst_i$
that are applied to $Q_i$.  Let $Q_i^{\ast}$ be a subregion of $Q_i$
that is obtained by applying all link instructions in
$\inst_Q \setminus \inst_i$ to $Q_i$.  Note that any part of
  $Q_i$ reallocated by $\inst_Q\setminus \inst_i$ lies within
$\segin_i \cup \segout_i$, and thus the cores remain unchanged.  We
show that $Q_i^\ast$ is well-defined, meaning that $Q_i^\ast$ is
invariant under the order in which the instructions in $\inst_Q$ are
applied.

\begin{lemma}\label{lem:layer.seg.preservation}
  Let $(I_1, I_2, k) \in \inst_i$ be a link instruction, and let $Z_1$
  and $Z_2$ be the layer segments in $L_k$ corresponding to $I_1$ and
  $I_2$, respectively.  Then, the layer segments $L_k[Z_1, Z_2]$
  remain assigned to $P_i$ under any link instruction in
  $\inst_Q \setminus \inst_i$.
\end{lemma}
\begin{proof}
  Let $J_1 \in \mathcal{J}_Q$ be the circular interval such that $J_1$ assigns the 
  layer $L_k$ to the link instruction $(I_1,I_2,k)$.
  Then, $(I_1,I_2,k) \in \inst(J_1)$.
  Suppose that there exists $J_2 \in \mathcal{J}_Q$ such that 
  some $\lambda_1 \in \inst(J_1)$ and $\lambda_2 \in \inst(J_2)$ reallocate the same layer segment $Z' \subseteq L_{k}$.

  By the rule of layer assignment, 
  the layers assigned to $\inst(J_1)$ and $\inst(J_2)$ are distinct 
  if $J_1$ and $J_2$ properly intersect. 
  Let $I'\in\mathcal{I}_{Q}$ be the circular interval corresponding to $Z'$. 
  Then, $I'$ is contained in $J_1 \cap J_2$. 
  Since $J_1$ and $J_2$ are interior-disjoint, 
  this can only happen if $I'$ is degenerate that lies at a shared endpoint of $J_1$ and $J_2$. 

  For $I_a^1, I_b^1 \in \mathcal{I}^+_{i_1}$, $J_1$ is the smallest circular interval that spans 
  from $I_a^1$ to $I_b^1$ in counterclockwise direction.
  Similarly, for $I_a^2, I_b^2 \in \mathcal{I}^+_{i_2}$, 
  $J_2$ is the smallest one that spans from $I_a^2$ to $I_b^2$ 
  in counterclockwise direction.
  Let $Z_a^\ell$ and $Z_b^\ell$ be layer segments in $L_{k}$ 
  corresponding to $I_a^\ell$ and $I_b^\ell$, respectively, for each $\ell = 1, 2$. 
  Then, link instructions in $\inst(J_1)$ reallocate $L_k[Z_a^1,Z_b^1]$ to $P_{i}$ and 
  link instructions in $\inst(J_2)$ reallocate $L_k[Z_a^2,Z_b^2]$ to $P_{i'}$ for some $i' \in [m]$.

  Without loss of generality, assume that $J_1$ and $J_2$ appear on $\mybd{Q}$ 
  in counterclockwise order. 
  The non-degenerate intervals $I_b^1$ and $I_a^2$ share an endpoint 
  at which the degenerate interval $I'$ lies.
  It follows that $Z'$ lies in $L_k(Z_b^1, Z_a^2)$, and thus is not reallocated by any link instruction in 
  either $\inst(J_1)$ or $\inst(J_2)$. 
\end{proof}
Lemma~\ref{lem:layer.seg.preservation} implies that no layer
  segment is reassigned by more than one link instruction.
As a consequence, we obtain the following corollary, which
  states that $Q^\ast_i$ is well-defined.

\begin{corollary}\label{cor:num.reallocate}
  In the reconfiguration of $\Pi[Q]$, each layer segment in $Q$ is
  reallocated at most once.
\end{corollary}
Corollary~\ref{cor:num.reallocate} ensures that each layer segment may
be reallocated to a different piece at most once, and no chains of
reallocations such as
$P_{i_1} \rightarrow P_{i_2} \rightarrow P_{i_3}$ with $i_1 \neq i_2$
and $i_2\neq i_3$ occur.

This lemma further implies that layer segments in $\segout_i$ are
preserved, while only those in $\segin_i$ are reallocated by
$\inst_Q \setminus \inst_i$.  Let $\seginast_i$ be the subset of
$\segin_i$ that consists of layer segments preserved under
$\inst_Q \setminus \inst_i$.  Then, $Q_i^\ast$ is the subpolygon that
is obtained from $Q_i$ by removing those layer segments in
$\segin_i \setminus \seginast_i$.  Then
$Q_i^\ast = \Bigl(\bigcup_{j \in [t]} \core{C_j}\Bigr) \cup
\seginast_i \cup \segout_i,$ where $\seginast_i\subseteq \segin_i$.
Figure~\ref{fig:applying.instructions}(b--c) illustrates the
construction of $Q_i^\ast$, in which only layer segments in $\segin_i$
are reallocated by $\inst_Q \setminus \inst_i$.

\subparagraph{Path-Connectivity of $Q_i^\ast$.}  To prove that
$Q_i^\ast$ forms a connected piece, it suffices to verify two types of
path-connectivity among its constituent parts, which must be preserved
during the reallocation induced by $\inst_Q \setminus \inst_i$.
\begin{itemize}\denseitems
\item Each layer segment in $\seginast_i \cup \segout_i$ is
  path-connected to some $\core{C_j}$ within $Q_i^\ast$ for
  $C_j \in R_i$.
\item The cores $\{\core{C_j} \mid C_j \in R_i\}$ are mutually
  path-connected within $Q_i^\ast$.
\end{itemize}
Here, two sets $A, B\subseteq X$ are said to be path-connected within
a region $X$ if there exists a path in $X$ joining some $a \in A$ and
$b \in B$.

\begin{figure}[!b]
  \centering
  \includegraphics[width=0.8\textwidth]{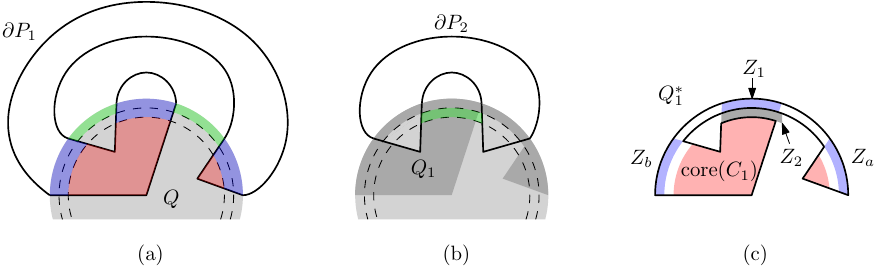}
  \caption{\small 
  (a) The piece $P_1$ defines the cores (red), $\segin_1$ (blue), and $\segout_1$ (green). 
  (b) The layer segment in $\segin_1$ is reallocated by $\inst_2$ which is defined on $P_2$. 
  (c) The layer segment $Z_1$ is connected to $\core{C_1}$ by a path within $Q_1^\ast$ that passes through $Z_b$.
  }
  \label{fig:segment.core.connect}
\end{figure}
\begin{lemma}\label{lem:segment.core.connect}
  Each layer segment in $\seginast_i\cup \segout_i$ is path-connected
  within $Q_i^*$ to the core of some $C_j \in R_i$.
\end{lemma}
\begin{proof}
  Let $Z_1$ be a layer segment of $C_{j}$ in $L_{\alpha}$, for $j \in [t]$ and $\alpha \in [h]$. 
  Let $I_1 \in \mathcal{I}_{ij}$ be a circular interval corresponding to $Z_1$.
  By Lemma~\ref{lem:layer.seg.preservation}, 
  every layer segment in $\segout_i$ is preserved under $\inst_Q \setminus \inst_i$, and 
  it is connected by a path within $Q^\ast_i$ to some layer segment in $\seginast_i$. 
  Thus, to show the path-connectivity of $Z_1$ to some core, 
  it suffices to consider the case that $Z_1 \in \seginast_i$.

  Assume that $Z_1$ is not path-connected to $\core{C_{j}}$ within $Q_i^\ast$.
  Then, there exists another layer segment $Z_2\in \segin_i \setminus \seginast_i$ 
  corresponding to $I_1$ such that $Z_2$ lies within $L_{\beta}$ for $\alpha < \beta < h$.
  That is, $Z_2$ is reallocated to a piece other than $P_i$, 
  thereby preventing a path between $Z_1$ and $\core{C_{j}}$. 
  
  Let $\lambda_2 \in \inst_Q \setminus \inst_i$ 
  be the link instruction that reallocates $Z_2$, and 
  let $J_2\in \mathcal{J}_Q$ be the circular interval that assigns the $\beta$-th layer 
  to $\lambda_2$. 
  Since $J_2$ has level $\beta$, 
  there exists $J_1 \in \mathcal{J}_Q$ with level $\alpha$ that contains $J_2$. 
  Note that $I_1 \subseteq J_2 \subseteq J_1$. 

  There exist $I_a, I_b \in \mathcal{I}^+_{Q}$ such that 
  $J_1$ is the smallest circular interval on $\mybd{Q}$ 
  spanning $I_a$ to $I_b$ in counterclockwise direction. 
  Let $Z_a$ and $Z_b$ be
  the layer segments in $L_{\alpha}$ corresponding to $I_a$ and $I_b$, respectively. 
  The link instructions $\inst(J_1)$ reallocates layer segments in $L_{\alpha}[Z_a, Z_b]$ to some piece. 
  Since $Z_1 \in L_{\alpha}[Z_a, Z_b]$ and $Z_1 \subseteq P_i$, this piece must be $P_i$.
  Moreover, we previously showed that every layer segment in $L_{\alpha}[Z_a, Z_b]$ 
  is preserved under $\inst_Q\setminus \inst_i$. 
  This implies that $Z_1$ is path-connected to both $Z_a$ and $Z_b$ in $Q^\ast_i$.
  Thus, the path-connectivity of $Z_1$ to a core within $Q^\ast_i$ can be reduced to that of $Z_b$ (or $Z_a$).

  Let $j' \in [t]$ such that $Z_b$ is the layer segment of $C_{j'}\in R_i$, 
  i.e., $I_b \in \mathcal{I}^+_{ij'}$.
  Since $J_1$ is the smallest circular interval in $\mathcal{J}_Q$ containing $I_a$ and $I_b$, 
  it follows that any proper subinterval $J \subsetneq J_1$ excludes $I_b$.
  That is, for all $k>\alpha$, 
  no link instruction in $\inst(J)$ reallocates 
  a layer segment corresponding to $I_b$ in $L_k$. 
  Thus, $Z_b$ is path-connected to the core of $C_{j'}$ within $Q_i^\ast$.
  Figure~\ref{fig:segment.core.connect} illustrates this situation: 
  although $Z_1$ is not directly connected to the core of $C_1$, 
  a detour path within $Q_1^\ast$ 
  exists from $Z_1$ to the core via $Z_b$.
\end{proof}

By Lemma~\ref{lem:segment.core.connect}, every layer segment
  in $\seginast_i \cup \segout_i$ has a path to some core within
  $Q_i^\ast$. It remains to show that
$\core{C_1}, \ldots, \core{C_t}$ are mutually path-connected within
$Q^\ast_i$.

\begin{lemma}\label{lem:core.mutual.connect}
  All cores of elements in $R_i$ are mutually path-connected within
  $Q_i^\ast$.
\end{lemma}
\begin{proof}
  Let $J_1 \in \mathcal{J}_Q$ be a circular interval 
  that is formed by a connected component $T$ in $G_{i}$. 
  The vertices in $T$ correspond to intervals $\{I_1, I_2,\ldots, I_{|T|}\} \subseteq \mathcal{I}^+_{i}$, 
  which are listed in counterclockwise order along $\mybd{Q}$.
  Then, $J_1$ is the smallest circular interval spanning $I_1$ to $I_{|T|}$ in counterclockwise direction. 
  The link instructions in $\inst(J_1)$ merge the corresponding elements $C_{\pi(1)}, \ldots, C_{{\pi(|T|)}}$ into a single piece, 
  where $\pi\colon [|T|]\to [t]$ is defined by $I_\ell \in \mathcal{I}^+_{i\pi(\ell)}$.

  Let $\alpha \in [h]$ be the layer index assigned to $\inst(J_1)$, and
  let $Z_\ell$ denote the layer segment in $L_{\alpha}$ corresponding to $I_\ell$ for each $\ell = 1,\ldots, |T|$.
  The link instructions in $\inst(J_1)$ reallocate 
  the layer segments $L_{\alpha}[Z_1, Z_{|T|}]$ to $P_{i}$. 
  By Lemma~\ref{lem:layer.seg.preservation}, the layer segments in $L_{\alpha}[Z_1, Z_{|T|}]$ 
  remain assigned to $P_i$ under the link instructions $\inst_Q \setminus \inst_i$. 
  We begin by showing that each $Z_\ell$ is path-connected to the corresponding $\core{C_{\pi(\ell}}$ within $Q^\ast_i$. 
  
  Let $J_2 \in \mathcal{J}_Q$ be the circular interval that is properly contained in $J_1$. 
  We show that $J_2$ must be interior-disjoint from each $I_j \in \{I_1,\ldots,I_{|T|}\}$.
  Suppose not, and there exists such a $J_2$ that contains some $I_{\ell'} \in \{I_1,\ldots,I_{|T|}\}$.
  Assume that $J_2$ is constructed in $G_{i'}$ for some $i' \in [m]$.
  Let $I_a$ and $I_b$ be two intervals in $\mathcal{I}^+_{i'}$ such that 
  $J_2$ is the smallest circular interval 
  spanning $I_a$ to $I_b$ in counterclockwise direction.
  Since $J_1 \neq J_2$, 
  $I_a \neq I_{\ell'}$ and $I_b \neq I_{\ell'}$.
  Then, $I_1$, $I_a$, $I_{\ell'}$, and $I_b$ 
  appear along $\mybd{Q}$ in counterclockwise order.
  
  By Lemma~\ref{lem:path.connected}, 
  there exist simple paths 
  $\gamma_1 \subseteq P_{i} \setminus \myint{Q}$ and 
  $\gamma_2 \subseteq P_{i'} \setminus \myint{Q}$
  such that 
  $\gamma_1$ connects an interior point of $I_1$ to an interior point of $I_{\ell'}$, and 
  $\gamma_2$ connects an interior point of $I_a$ to an interior point of $I_b$.
  Each $\gamma_r$ can be continuously deformed into 
  a simple path $\tilde{\gamma}_r \subseteq \mybd{Q}$ for $r = 1,2$. 
  Observe that 
  $\tilde{\gamma}_1$ traverses the arc from $I_1$ to $I_{\ell'}$, and 
  $\tilde{\gamma_2}$ traverses the arc from $I_a$ to $I_b$, both in counterclockwise order.
  Given the cyclic order $I_1,I_a,I_{\ell'}, I_b$ along $\mybd{Q}$, 
  $\tilde{\gamma}_1$ and $\tilde{\gamma}_2$ must properly intersect.
  Thus, their original paths $\gamma_1$ and $\gamma_2$ intersect at a common interior point $x$. 

  Consider the case that $i \neq i'$.
  Then $x$ lies in $\mybd{P_{i}} \cap \mybd{P_{i'}}$ 
  since $P_{i}$ and $P_{i'}$ are interior-disjoint. 
  A sufficiently small ball $B$ centered at $x$ is split by $\gamma_1$ 
  into two connected regions, 
  at least one of which lies entirely within $P_{i}$. 
  Moreover, $\gamma_2$ also passes through $x$ and crosses $\gamma_1$ at that point. 
  Thus, $\gamma_2$ must enter the interior of $P_{i}$, contradicting that $\gamma_2 \subseteq P_{i'}$.  

  In the case that $i = i'$, 
  we obtain a new simple path $\gamma'$ by concatenating subpaths of $\gamma_1$ and $\gamma_2$ at $x$. 
  This path lies in $P_{i} \setminus \myint{Q}$, 
  and connects an interior point of $I_1$ to an interior point of $I_a$.
  By Lemma~\ref{lem:path.connected}, 
  the vertices corresponding to $I_1$ and $I_a$ are connected in 
  the underlying graph of $G_i$, 
  which implies $J_1 = J_2$, a contradiction.
  Therefore, for any $J\in \mathcal{J}_Q$ with $J\subsetneq J_1$, 
  $J$ and $I_\ell$ are interior-disjoint for all $\ell =1 ,\ldots, |T|$.
  
  Now let $Z'$ be an arbitrary layer segment in $L_k$ corresponding to $I_\ell$, 
  for some $\ell \in [|T|]$ and some $\alpha < k < \nlayers$.
  By the previous argument, no 
  $J \in \mathcal{J}_Q$ with $J\subsetneq J_1$ intersects $\myint{I_\ell}$, and 
  hence $Z'$ is not reallocated to any piece other than $P_i$.  
  Since $\ell$ was arbitrary, 
  it follows that the layer segment $Z_\ell$ in $L_{\alpha}$ 
  is path-connected to $\core{C_{\pi(\ell)}}$ within $Q^\ast_i$,
  for all $\ell = 1, \ldots, |T|$. 

  The layer segments $\{Z_1, \ldots, Z_{|T|}\}$ are merged by
  the link instructions in $\inst(J_1)$, 
  and Lemma~\ref{lem:layer.seg.preservation} ensures that 
  this connectivity is preserved under the reconfiguration.
  As each $Z_\ell$ is path-connected to $\core{C_{\pi(\ell)}}$ within $Q_i^\ast$,
  all cores $\core{C_{\pi(1)}}, \ldots, \core{C_{{\pi(|T|)}}}$ are mutually path-connected in $Q_i^\ast$.
  We extend the argument to every link instruction in $\inst_i$.
  The cores of all elements in $R_i$ are connected within $Q_i$, 
  and this connectivity is preserved under the reconfiguration. 
\end{proof}

By Lemmas~\ref{lem:segment.core.connect}
and~\ref{lem:core.mutual.connect}, each $Q^\ast_i$ is connected.
Thus, applying $\inst_Q$ to $\Pi[Q]$ yields the partition
$\Pi^\ast[Q] = \{Q^\ast_1, Q^\ast_2, \ldots, Q^\ast_m\}$, where each
$Q_i^\ast$ is a connected subregion of $Q$.

\subparagraph{Remarks.}  Applying link instructions in $\inst_Q$ may
induce holes within merged pieces in the reconfigured partition of
$Q$.  When $Q^\ast_i$ contains holes, reallocating them to $P_i$ does
not increase $\dwidth{\vv{v}}{Q^\ast_i}$ for every $\vv{v}\in \usetp$.
Therefore, each $Q^\ast_i$ can be regarded as a simple polygon.

\subsection{Feasibility of reconfigured partitions}
We first observe that the reconfigured partition
  $\Pi^\ast[Q] = \{Q^\ast_1, Q^\ast_2, \ldots, Q^\ast_m\}$ remains as a
  valid partition of $Q$. By construction, the cut
  constraint $U$ is also preserved: each layer segment is bounded by
the boundaries of layers and the cuts from $\Pi[Q]$, all aligned with
  directions in $U$. It remains to check the unit-width
  constraint $W$.

Since $\Pi= \{P_1,\ldots, P_m\}$ is a solution to $\prob{P,W,U}$,
there exists a vector $\vv{v_i} \in W$ such that
$\dwidth{\vv{v_i}}{P_i} \le 1$.  In other words, there exists a unit
strip $H_i$ with normal vector $\vv{v_i}$ that contains $P_i$.  If
every layer segment reallocated by $\inst_i$ is contained within
$H_i$, then $Q^\ast_i$ also satisfies unit-width constraint $W$.

Assume that $\lambda=(I_1,I_2,k')$ is a link instruction associated
with a directed edge from $v_1$ to $v_2$ in $G_i$ for
$I_1, I_2 \in \mathcal{I}^+_i$ and $k' \in [h]$.  Let $j, j' \in [t]$
be two distinct indices such that $I_1 \in \mathcal{I}^+_{ij}$ and
$I_2 \in \mathcal{I}^+_{ij'}$.  Recall that the edge between $v_1$ and
$v_2$ is added to $G_i$ if and only if there is a transition between
$\mybd{C_{j}}$ and $\mybd{C_{j'}}$ during the counterclockwise
traversal on $\mybd{P_i}$.

For each $k \in [h]$, let $Z_1^k$ and $Z_2^k$ denote the layer
segments in $L_k$ corresponding to $I_1$ and $I_2$,
respectively.  The link instruction $\lambda$ reallocates the layer
segments in $L_{k'}(Z_1^{k'}, Z_2^{k'})$ to $P_i$, where $Z_1^{k'}$
and $Z_2^{k'}$ are already assigned to $P_i$ in $\Pi[Q]$.  We define
$\segunion_\lambda \coloneqq \bigcup_{k \in [h]}\bigcup_{Z \in
  \mathcal{Z}^k}Z,$ where
$\mathcal{Z}^k\coloneqq L_k(Z_1^{k}, Z_2^{k}).$ It suffices to
  show that $\segunion_\lambda \subseteq H_i$.  If this inclusion
  holds, every layer segment reallocated by $\lambda$ lies within
  $H_i$, and thus the unit-width constraint $W$ is preserved.

\subparagraph{Decomposition of $\mybd{\segunion_\lambda}$.}  The
region $\segunion_\lambda$ is connected and bounded by four parts: two
continuous portions of the inner and outer boundaries of $\corridor$,
and two subsegments of edges from $C_{j}$ and $C_{j'}$.  The outer
boundary of $\corridor$ is $\mybd{Q}$ and its inner boundary is
$\mybd{Q^\phi}$, where $Q^\phi$ is the inner $\phi$-offset polygon of
$Q$.

For the sake of clarity, we introduce interval notation to represent
subpaths of $\mybd{Q}$ and $\mybd{Q}^\phi$.  For any two points
$x,y \in \mybd{Q}$, we denote by $\mybd{Q}[x,y]$ the portion of
$\mybd{Q}$ from $x$ to $y$ in counterclockwise order, including both
endpoints.  Let $\mybd{Q}(x,y] = \mybd{Q}[x,y]\setminus\{x\}$,
$\mybd{Q}[x,y) = \mybd{Q}[x,y]\setminus\{y\}$, and
$\mybd{Q}(x,y) = \mybd{Q}[x,y]\setminus\{x,y\}$.  Similarly, we define
$\mybd{Q}^\phi[x,y], \mybd{Q}^\phi(x,y], \mybd{Q}^\phi[x,y)$, and
$\mybd{Q}^\phi(x,y)$ as portions of $\mybd{Q}^\phi$.

Since $I_1$ and $I_2$ are non-degenerate intervals on $\mybd{Q}$, let
$I_1 = \mybd{Q}[p_1,q_1]$ and $I_2 = \mybd{Q}[p_2,q_2]$ for some
points $p_1,q_1,p_2,q_2 \in \mybd{Q}$ with $p_1 \neq q_1$ and
$p_2 \neq q_2$. The portion of $\mybd{\segunion_\lambda}$
  contained in $\mybd{Q}$ is $\mybd{Q}[q_1,p_2]$. Let
  $r_1, r_2\in \mybd{Q}^\phi$ be the endpoints of the portion of
  $\mybd{\segunion_\lambda}$ lying on $\mybd{Q}^\phi$, which we denote
  by $\mybd{Q}^\phi[r_1, r_2]$. Finally, the parts of
$\mybd{\segunion_\lambda}$ along the edges of $C_{j}$ and $C_{j'}$
correspond to the segments $\overline{q_1r_1}$ and
$\overline{p_2r_2}$, respectively.  Thus, $\mybd{\segunion_\lambda}$
decomposes into the four parts
$\mybd{Q}[q_1,p_2],\mybd{Q}^\phi[r_1,r_2],\overline{q_1r_1}$, and
$\overline{p_2r_2}$.  This decomposition is shown in
Figure~\ref{fig:Z_lambda}(a).

Since $H_i$ is convex, $\segunion_\lambda \subseteq H_i$ if and only
if all parts of $\mybd{\segunion_\lambda}$ are contained in $H_i$.
Note that $\overline{q_1r_1}$ and $\overline{p_2r_2}$ lie in $H_i$, as
both $C_{j}$ and $C_{j'}$ are contained in $P_i$.  To prove
$\segunion_\lambda \subseteq H_i$, it remains to show that
$\mybd{Q}[q_1,p_2]$ and $\mybd{Q}^\phi[r_1,r_2]$ are contained in
$H_i$.

\begin{figure}[!t]
  \centering
  \includegraphics[width=0.6\textwidth]{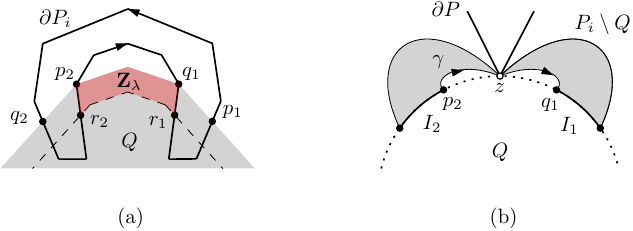}
  \caption{\small (a) The link instruction $\lambda = (I_1,I_2,k)$ is
    derived from the traversal of $\mybd{P_i}$ from $p_1$ to $q_2$ or
    from $p_2$ to $q_1$, where $I_1 = \mybd{Q}[p_1,q_1]$ and
    $I_2 = \mybd{Q}[p_2,q_2]$.  $\mybd{\segunion_\lambda}$ consists of
    $\mybd{Q}[q_1,p_2],\mybd{Q}^\phi[r_1,r_2],\overline{q_1r_1}$, and
    $\overline{p_2r_2}$.  (b) The path $\gamma$ from $p_2$ to $q_1$
    must pass through $z$ when $\mybd{P}$ intersects $\mybd{Q}$ at
    $z$.  }
  \label{fig:Z_lambda}
\end{figure}

\subparagraph{Containment of $\mybd{Q}[q_1, p_2]$ in $\myint{P}$.}
During the transition between $\mybd{C_{j}}$ and $\mybd{C_{j'}}$ in
the counterclockwise traversal on $\mybd{P_i}$, it follows a simple
path, denoted by $\gamma$, which lies outside $\myint{Q}$.  The path
$\gamma$ exits and re-enters $\myint{Q}$ through the endpoints of
$I_1$ and $I_2$.

Let $v_1$ and $v_2$ denote the vertices in $G_i$ corresponding to
$I_1$ and $I_2$, respectively.  The link instruction
$\lambda= (I_1, I_2, k')$ is derived from the directed edge
$(v_1,v_2)$ in $G_i$.  The direction of the edge is determined by how
$\gamma$ winds around $\mybd{Q}$ (either clockwise or
counterclockwise) and whether the path exits $\myint{Q}$ through the
endpoint of $I_1$ or that of $I_2$.  If $\gamma$ winds around
$\mybd{Q}$ counterclockwise, it exits $\myint{Q}$ from $p_1$ and
re-enters at $q_2$.  Otherwise, it exits $\myint{Q}$ from $p_2$ and
re-enters at $q_1$.  Figure~\ref{fig:Z_lambda}(a) illustrates both
cases.

Up to this point, we consider both clockwise- and
counterclockwise-type instructions.  However, counterclockwise ones
can be omitted in the reconfiguration.  Assume that $\gamma$ winds
around $\mybd{Q}$ counterclockwise.  By
Corollary~\ref{cor:unique.path.structure}, there exists a unique
component $X \in \mathcal{X}_i$ whose boundary contains $I_1$, $I_2$,
and the path $\gamma$.  Since $X$ is a weakly simple polygon, we can
traverse $\mybd{X}$ in counterclockwise order.  This traversal
encounters a sequence of circular intervals in $\mathcal{I}^+_i$, and,
from the construction of $G_i$, each consecutive pair of intervals in
this sequence corresponds to an edge in $G_i$ whose direction is
determined by whether the subpath of $\mybd{X}$ between the
  intervals winds around $\mybd{Q}$ clockwise or counterclockwise.

Since the traversal of $\mybd{X}$ follows the path $\gamma$ from $I_1$
to $I_2$ that winds around $\mybd{Q}$ counterclockwise,
the other subpath of $\mybd{X}$ runs from $I_2$ back to $I_1$,
  and winds around $\mybd{Q}$ clockwise. This path encounters a
sequence of intervals in $\mathcal{I}^+_i$ starting from $I_2$ to
$I_1$, where each consecutive pair of intervals induces a
clockwise-type edge in $G_i$.  In Figure~\ref{fig:path.homotopy}(b),
the counterclockwise-type path from $I_4$ to $I_1$ corresponds to the
sequence of clockwise-type paths $I_1\rightarrow I_5$ and
$I_5 \rightarrow I_4$.  Thus, link instructions associated with
counterclockwise-type edges can be omitted without affecting the
resulting partition $\Pi^\ast[Q]$.

Without loss of generality, we restrict our analysis to clockwise type
instructions.  In this case, $\gamma$ is a path from $p_2$ to $q_1$.
The path $\gamma$ can be continuously deformed into a path
$\tilde{\gamma}= \mybd{Q}[q_1, p_2]$ on $\mybd{Q}$ while preserving
its endpoints.  As illustrated in Figure~\ref{fig:Z_lambda}(b), if
$\mybd{P}$ intersects $Q$ at a point $z \in \mybd{Q}[q_1,p_2]$, then
$\gamma$ must pass through $z$.  This implies that $\mybd{P_i}$
intersects itself at $z$, contradicting that $P_i$ is a simple
polygon.

\begin{lemma}\label{lem:lying.interior.P}
  All points on $\mybd{Q}[q_1, p_2]$ lie within the interior of $P$.
\end{lemma}
\begin{proof}
  Suppose, for the sake of contradiction, 
  that there is a point $z \in \mybd{Q}[q_1, p_2]$ 
  that touches $\mybd{P}$. 
  Then, the point $z$ lies in $\mybd{Q} \cap \mybd{P}$. 
  Note that the path $\gamma\subseteq P \setminus \myint{Q}$ is a simple path connecting 
  $q_1$ to $p_2$, and $\tilde{\gamma}$ is the path homotopic to $\gamma$ that follows the arc $\mybd{Q}[q_1,p_2]$ from $p_2$ to $q_1$.
  Then $\gamma$ passes through $z$ as shown in Figure~\ref{fig:Z_lambda}(b).

  Consider a sufficiently small open ball $B$ centered at $z$. 
  The path $\gamma$ intersects $B$ and separates it into two open regions $B_1$ and $B_2$:
  $B_1$ lies in $\myint{P_i}$, and $B_2$ lies outside $P_i$.  
  Since $\mybd{P_i}$ is traversed in counterclockwise direction, 
  $\myint{P_i}$ lies to the left of an observer walking along $\gamma$.
  Then, $B_2$ intersects $\myint{Q}$, while $B_1$ does not.  
  The only way for $\mybd{P}$ to reach $z$ is from $B_1$.  
  Since $B_1 \subset \myint{P_i}$, 
  this contradicts that $P_i$ is a simple polygon.
\end{proof}

Note that a point $p \in \myint{P}$ if and only if there exists a
sufficiently small ball $B(p)$ centered at $p$ such that
$B(p) \subseteq \myint{P}$.  We slightly extend the polygonal chain
$\mybd{Q}[q_1, p_2]$ beyond its endpoints and denote the resulting
chain by $\mybd{Q}[q_1^-, p_2^+]$ for $q_1^- \in \mybd{Q}(p_1,q_1)$
and $p_2^+ \in \mybd{Q}(p_2,q_2)$.  Since
$\mybd{Q}[q_1, p_2] \subseteq \myint{P}$ by
Lemma~\ref{lem:lying.interior.P}, $\mybd{Q}[q_1^-, p_2^+]$ also lies
in $\myint{P}$.  Furthermore, as $I_1$ and $I_2$ are non-degenerate,
the slight extension guarantees that
$\mybd{Q}[q_1^-, p_2^+] \subseteq \mybd{Q}[p_1,q_2]$.

\begin{figure}[!b]
  \centering
  \includegraphics[width=0.7\textwidth]{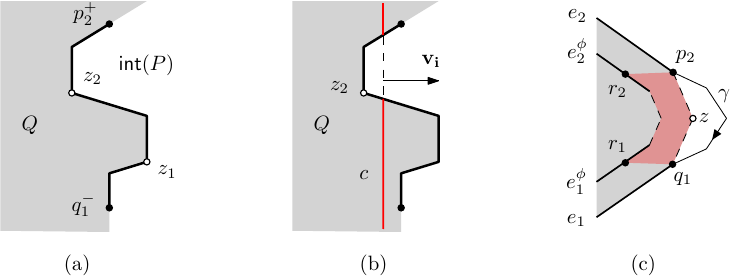}
  \caption{\small (a) The turning points on $\mybd{Q}[q_1^-, p_2^+]$
    are $z_1$ and $z_2$.  (b) For the guillotine cut $c$ along
    $ x= x(z) + \varepsilon$, the intersection $c \cap Q$ consists of
    at least two maximal line segments.  (c) $\gamma$ contains a point
    with $x$-coordinate larger than that of the turning point $z$.
    The sequence of edges on $\mybd{Q}^\phi[r_1,r_2]$ from $e_1^\phi$
    to $e_2^\phi$ corresponds to a subsequence of those of
    $\mybd{Q}[p_1,q_2]$.  }
  \label{fig:turning.points}
\end{figure}

\subparagraph{Turning points on $\mybd{Q}[q_1^-, p_2^+]$.}  Without
loss of generality, assume $H_i$ is a vertical strip, meaning
that its boundary consists of two vertical lines.  A portion of
$\mybd{Q}$ is a polygonal chain, consisting of a sequence of line
segments.  As we traverse a polygonal chain from one endpoint to the
other, these line segments are encountered sequentially.  A point
$z \in \mybd{Q}(q_1^-, p_2^+)$ is called a \emph{turning point} of
$\mybd{Q}[q_1^-, p_2^+]$ if the traversal changes its
horizontal direction (from leftward to rightward or vice
  versa) at $z$. Since $q_1^-$ and $p_2^+$ lie sufficiently close to
$q_1$ and $p_2$ along $\mybd{Q}$, every turning must occur on
$\mybd{Q}[q_1, p_2]$.

When $\mybd{Q}[q_1^-, p_2^+]$ contains vertical edges, the above
definition of turning points is not sufficient.  Consider a path that
initially moves in the positive (or negative) $x$-direction, then
follows a vertical segment, and subsequently moves in the negative (or
positive) $x$-direction.  We define the lowest point on the vertical
segment as the unique turning point on that segment.  See
Figure~\ref{fig:turning.points}(a).

Note that the number of turning points is invariant under the
traversal direction; that is, it remains unchanged whether we traverse
$\mybd{Q}[q_1^-, p_2^+]$ from $q_1^-$ to $p_2^+$ or in the reverse
direction.  As illustrated in Figure~\ref{fig:turning.points}(b), if
there are two turning points on $\mybd{Q}[q_1^-, p_2^+]$, we can draw
a vertical guillotine cut $c$ in $P$ such that $c \cap Q$ is
disconnected, since all points on $\mybd{Q}[q_1^-, p_2^+]$ lie within
$\myint{P}$.  This implies that $Q$ cannot be $\overline{W}$-convex
with respect to $P$.

\begin{lemma}\label{lem:num.turning.points}
  If $Q$ is $\overline{W}$-convex with respect to $P$, then the number
  of turning points on $\mybd{Q}[q_1^-, p_2^+]$ is at most one.
\end{lemma}
\begin{proof}
  Suppose that $\mybd{Q}[q_1^-, p_2^+]$ has at least two turning points when $Q$ is $\overline{W}$-convex with respect to $P$.
  Let $z_1$ and $z_2$ be two consecutive turning points along the traversal from $q_1^-$ to $p_2^+$;
  no other turning point lies on $\mybd{Q}(z_1,z_2)$. 
  Without loss of generality, 
  assume that the traversal moves in the positive $x$-direction immediately before $z_1$ and 
  moves in the negative $x$-direction immediately before $z_2$. 
  At least one of these points must be a reflex vertex of $Q$ 
  as it corresponds to a right turn in the counterclockwise traversal of $\mybd{Q}$.
  We may assume that the reflex vertex of $Q$ lies at $z_2$. 
  See Figure~\ref{fig:turning.points}(a).

  Using Lemma~\ref{lem:lying.interior.P}, we have $\mybd{Q}[q_1^-, p_2^+] \subseteq \myint{P}$, and then 
  the turning point $z_2$ also lies in $\myint{P}$. 
  For an infinitesimally small $\varepsilon > 0$, 
  let $\ell$ be the vertical line defined by $x = x(z_2) + \varepsilon$.
  Since $\varepsilon$ is chosen infinitesimally small, 
  $c$ intersects $\mybd{Q}$ in at least four distinct points.
  Thus, $c \cap Q$ is disconnected.
  Recall that the coordinate system has been rotated so that the vector $\vv{v_i} \in W$ 
  becomes horizontal.
  Since $c \cap Q$ is disconnected, $Q$ is not $\overline{W}$-convex with respect to $P$, a contradiction.
\end{proof}

Since both $\mybd{Q}[p_1,q_1]$ and $\mybd{Q}[p_2,q_2]$ lie in $P_i$,
$q_1^-, p_2^+ \in H_i$.  If there are no turning points on
$\mybd{Q}[q_1^-, p_2^+]$, then the points with the largest and
smallest $x$-coordinates along $\mybd{Q}[q_1^-, p_2^+]$ appear at
$q_1^-$ and $p_2^+$.  Since $H_i$ is a vertical slab, all points on
$\mybd{Q}[q_1^-, p_2^+]$ lie in $H_i$.

Consider the case that $\mybd{Q}[q_1^-, p_2^+]$ has a turning point at
$z \in \mybd{Q}[q_1, p_2]$, and it is unique by
Lemma~\ref{lem:num.turning.points}.  Then, the point with the largest
or smallest $x$-coordinate along $\mybd{Q}[q_1^-, p_2^+]$ may appear
at $z$.  Without loss of generality, we assume that $z$ is the point
with the largest $x$-coordinate.  The point with the smallest
$x$-coordinate lies at $q_1^-\in H_i$ or $p_2^+\in H_i$.

Recall that the path $\gamma$ follows from $p_2$ to $q_1$ outside
$\myint{Q}$ and is deformed into the path $\tilde{\gamma}$ that
traverses $\mybd{Q}[q_1,p_2]$.  Note that $z$ is a convex vertex of
$Q$ with locally largest $x$-coordinate, and $\gamma$ encloses $z$
from outside $\myint{Q}$.  It follows that $\gamma$ must pass through
a point with $x$-coordinate at least $x(z)$.  Since
$\gamma \subseteq P_i \subseteq H_i$, the $x$-coordinate of $z$ is
smaller than that of the right boundary of $H_i$.  Thus,
$\mybd{Q}[q_1^-, p_2^+]\subseteq H_i$. See
Figure~\ref{fig:turning.points}(c).

\subparagraph{Containment of $\segunion_\lambda$ within $H_i$.}
Revisiting the boundary of $\segunion_\lambda$, we have shown that
three parts, $\overline{r_1q_1}, \overline{r_2p_2}$, and
$\mybd{Q}[q_1,p_2]$, are contained in $H_i$.  The remaining part is
$\mybd{Q^\phi}[r_1,r_2]$ which is the portion of the inner
$\phi$-offset polygon of $Q$.  By definition of the offset polygon,
each edge of $Q^\phi$ is parallel to its corresponding edge in $Q$ and
the edges of $Q^\phi$ appear in the same cyclic order along its
boundary as the edges of $Q$.

We traverse the chain $\mybd{Q}[q_1^-, p_2^+]$ from $q_1^-$ to
$p_2^+$.  Let $e_1$ and $e_2$ be the edges of $Q$ that contain the
first and last segments of this chain, respectively.  Likewise,
traversing $\mybd{Q}^\phi[r_1,r_2]$ from $r_1$ to $r_2$ gives edges
$e_1^\phi, e_2^\phi$ of $Q^\phi$ incident to $r_1$ and $r_2$,
respectively.  We claim that the edges of $Q$ corresponding to
$e_1^\phi$ and $e_2^\phi$ appear along the traversal of
$\mybd{Q}[q_1^-, p_2^+]$.

The segments $\overline{r_1q_1}$ and $\overline{r_2p_2}$ lie on edges
of $C_{j}$ and $C_{j'}$, respectively.  Since $\phi < \phi_{ij}$ and
$\phi < \phi_{ij'}$, the edges of $Q$ corresponding to $e_1^\phi$ and
$e_2^\phi$ are incident to $q_1$ and $p_2$, respectively.
Recall that we work on the extended chain
  $\mybd{Q}[q_1^-, p_2^+]$, obtained by slightly extending
  $\mybd{Q}[q_1,p_2]$.  This guarantees that the corresponding edges
  of $Q$ appear along $\mybd{Q}[q_1^-, p_2^+]$. Consequently, the
sequence of line segments forming $\mybd{Q^\phi}[r_1, r_2]$
corresponds to a subsequence of those forming
$\mybd{Q}[q_1^-, p_2^+]$.  See Figure~\ref{fig:turning.points}(c) for
an illustration of this correspondence.  Thus, by
Lemma~\ref{lem:num.turning.points}, the number of turning points of
$\mybd{Q^\phi}[r_1, r_2]$ is also at most one.

Assuming that the largest $x$-coordinate of $\mybd{Q}[q_1^-, p_2^+]$
occurs at its turning point, the largest $x$-coordinate of
$\mybd{Q^\phi}[r_1, r_2]$ is smaller than that of
$\mybd{Q}[q_1^-, p_2^+]$. The argument is symmetric when the
  smallest $x$-coordinate is attained at the turning point. It
follows that $\mybd{Q^\phi}[r_1, r_2]$ is also contained in $H_i$,
which completes the proof that
$\mybd{\segunion_\lambda}\subseteq H_i$.  Hence, the reconfigured
partition $\Pi^\ast[Q] = \{Q_1^\ast, \ldots, Q_m^\ast\}$ is a solution
to the problem $\prob{Q, W, U}$ with at most $m$ connected pieces;
thus, by definition, $\nopt{Q} \le m = \nopt{P}$.

\section{Bang-type theorem for partitions of a convex body}\label{sec:bang.type.theorem}
We adapt the reconfiguration technique in
Section~\ref{sec:reconfig.restricted.partition} to prove
Theorem~\ref{thm:bang.type.analogue}.  We then show that, when
$\overline{W} \subseteq U$, an optimal partition of a convex polygon
$P$ is achieved by equally spaced parallel cuts, which can be computed
in linear time.

Let $K$ be a convex body in $\mathbb{R}^2$, and let
$P_1\cup P_2 \cup \cdots \cup P_m$ be its arbitrary partition.  Note
that each $P_i$ is compact and possibly non-convex.  Let $\conv{X}$
denote a convex hull of a set $X$ in $\mathbb{R}^2$.  A pocket of
$\conv{X}$ is defined as a closure of a connected component of
$\conv{X}\setminus X$.  Each pocket is bounded by a subpath along
$\mybd{P_i}$ and a unique line segment lying outside $P_i$.  We refer
to this segment as the \emph{hull-edge} of the pocket.

To convexify $P_i$, we iteratively reallocate its pockets to $P_i$.
However, such a reallocation may split other pieces $P_j$ with
$j \neq i$.  Accordingly, we perform a reconfiguration step to merge
such fragments into a single piece in $K$, ensuring that each piece
remains connected.

This configuration mirrors
Section~\ref{sec:reconfig.restricted.partition}, but with spatial
roles reversed.  Previously, we considered the restriction of a
partition to a subpolygon $Q$, and reconnected the fragments of other
pieces within $Q$.  Here, we restrict the partition to the complement
of a pocket, and consider the fragments of other pieces that lie
outside the pocket.  The circular intervals in the earlier setting now
correspond to the intervals along the hull-edge of the pocket.

\begin{lemma}\label{lem:convert.convex.partition}
  Let $\{P_1,\ldots,P_m\}$ be a partition of a convex body
  $K\subseteq \mathbb{R}^2$.  Then, there exists a convex partition
  $K = P_1^\ast \cup \cdots \cup P_m^\ast$, such that
  $\dwidth{\vv{v}}{P_i} \ge \dwidth{\vv{v}}{P_i^\ast}$ for all
  $\vv{v} \in \usetp$ and $i \in [m]$.
\end{lemma}
\begin{proof}
Let $\Delta$ denote a pocket of $\conv{P_i}$ bounded by hull-edge $\ell$, that will be reallocated to $P_i$. 
We define layers for reconfiguration 
by drawing a sequence of lines parallel to $\ell$, 
starting from $\ell$ and extending outward in the exterior direction of $\Delta$.  
These lines are placed at an infinitesimal distance apart, forming a narrow corridor along $\ell$.

For each other piece $P_j$ with $i\neq j$, 
we traverse $\mybd{P_j}$ 
and construct an undirected graph $G_j' = (V_j', E_j')$. 
Each vertex in $V_j'$ represents a connected component of $P_j \cap \ell$ with positive length, 
viewed as a non-degenerate interval along $\ell$. 
Two vertices are connected by an edge if a transition between their corresponding intervals occurs during the traversal.
This graph encodes the link instructions that reconnect 
the fragments of $P_j$ induced by the reallocation of $\Delta$. 
Let $\inst_\Delta$ denote the set of link instructions that are encoded in $G'_j$ for $j\neq i$.


For each connected component of $G_j'$, 
we take the smallest interval on $\ell$ containing all associated intervals.
This yields a set of intervals $\mathcal{J}'_j$ for each $j\neq i$. 
Let $\mathcal{J}'$ denote the union of all such sets, i.e., $\mathcal{J}' = \bigcup_{j \neq i} \mathcal{J}_j'$. 
This collection $\mathcal{J}$ 
is analogous to the set $\mathcal{J}$ from Section~\ref{sec:reconfig.restricted.partition} 
in that its intervals do not properly intersect each other. 
As before, we construct the transitive reduction $\trg{G}_{\mathcal{J}'}$ and 
assign layers to link instructions in $\inst_\Delta$ 
according to the depths of their corresponding nodes in $\trg{G}_{\mathcal{J}'}$.

Executing all link instructions in $\inst_\Delta$
merges the fragments of each $P_j$ in $K \setminus \Delta$ 
into a single piece. 
Iterating this process over all pockets of $\conv{P_i}$ eventually transforms each $P_i$ into its convex hull $\conv{P_i}$.
Throughout, the connectivity of all other pieces is preserved, 
so the number of connected pieces in the partition of $K$ remains unchanged.
The same procedure is applied to the other non-convex pieces. 

During the process, the convexity of initially convex pieces is preserved. 
Consider a convex piece $P_j$ that intersects a pocket $\Delta$ of $\conv{P_i}$ for $i \neq j$.
The intersection $P_j \cap \ell$ is a line segment, and the remaining part 
$P_j \setminus \Delta$ is also convex.  
Moreover, $P_j \setminus \Delta$ meets each layer $L_k$ in at most one connected component $Z_k$, which we refer to as the layer segment. 
By construction, 
the layer segments $Z_1,Z_2, \ldots ,Z_k$ are arranged consecutively along the corridor, so that 
their removal from $P_j \setminus \Delta$ 
is equivalent to intersecting with a single half-plane. 
This operation preserves convexity.
As a result, we obtain a convex partition of $K$, 
$\{P_1^\ast,P_2^\ast, \ldots, P_m^\ast\}$, where 
each $P_i$ is reconfigured to $P_i^\ast$.

Finally, 
we show that the directional width of each piece is non-increasing 
during reconfiguration:
for any $\vv{v} \in \usetp$ and any $i \in [m]$, we have
$\dwidth{\vv{v}}{P_i}\ge \dwidth{\vv{v}}{P_i^\ast}$.
Assume that two fragments $C_1$ and $C_2$ of some piece $P_{i'}$ are merged 
by a link instruction. 
The link instruction reallocates a region $\mathbf{Z}$ within a single layer. 
Since layer is defined as a strip bounded by two parallel lines, 
$\mathbf{Z}$ is bounded by an edge of $C_1$, 
an edge of $C_2$, and those two parallel lines. 
Hence, $\mathbf{Z} \subseteq \conv{C_1 \cup C_2} \subseteq \conv{P_{i'}}$. 
Therefore, 
the reconfigured piece $P_{i'}^\ast$ is contained in $\conv{P_{i'}}$, and 
for all $\vv{v} \in \usetp$, $\dwidth{\vv{v}}{P_{i'}}\ge \dwidth{\vv{v}}{P_{i'}^\ast}$.
\end{proof}

By Lemma~\ref{lem:convert.convex.partition}, we have a convex
partition $K = P_1^\ast\cup \cdots \cup P_m^\ast$ such that
$\dwidth{\vv{v}}{P_i} \ge \dwidth{\vv{v}}{P_i^\ast}$ for all
$\vv{v} \in \usetp$ and $i \in [m]$.  Given this convex partition,
Akopyan~\cite{Akopyan2019} showed that
$\sum_{i=1}^m r_K(P_i^\ast) \ge 1$, where
$r_K(P^\ast_i)= \sup\{h\ge 0 \mid \exists t \in \mathbb{R}^2 \
\text{such that}\ hK + t \subseteq P^\ast_i\}$.  For any direction
$\vv{v}\in \usetp$, we have
$r_K(P^\ast_i) \le \dwidth{\vv{v}}{P_i^\ast}/\dwidth{\vv{v}}{K}$.
Thus, for any subset $W \subseteq \usetp$,
$ \sum_{i=1}^m \inf_{\vv{v} \in
  W}\frac{\dwidth{\vv{v}}{P_i}}{\dwidth{\vv{v}}{K}} \ge \sum_{i=1}^m
\inf_{\vv{v} \in
  W}\frac{\dwidth{\vv{v}}{P^\ast_i}}{\dwidth{\vv{v}}{K}} \ge
\sum_{i=1}^m r_K(P^\ast_i) \ge 1$.

\subparagraph{Optimal partition for a convex polygon.}  Let $P$ be a
convex polygon with $n$ vertices, and let $W,U\subseteq \usetp$ such
that $\overline{W} \subseteq U$.  Choose an arbitrary vector
$\vv{u} \in W$; without loss of generality, assume that
$\vv{u} = (1,0)$.  Let $q$ be the leftmost vertex of $P$, and
partition $P$ by vertical lines along $x_i = x(q)+ i$ for all
$i = 1, \ldots, \lceil\dwidth{\vv{u}}{P}\rceil-1$.  This partitions
$P$ into $\lceil\dwidth{\vv{u}}{P}\rceil$ pieces, each of horizontal
width at most 1, and it is a feasible solution to $\prob{P, W, U}$.
Since $\vv{u}$ is chosen arbitrarily,
$\nopt{P} \le \min_{\vv{v}\in W} \lceil\dwidth{\vv{v}}{P}\rceil$.

Suppose, for the sake of contradiction, that the optimal partition has
fewer than $\min_{\vv{v}\in W}\lceil\dwidth{\vv{v}}{P}\rceil$ pieces.
Let $P = P_1 \cup \cdots \cup P_m$ be an optimal partition for
$\prob{P,W,U}$, with $m = \nopt{P}$.  By
Theorem~\ref{thm:bang.type.analogue}, we have
$1 \le \sum_{i=1}^m \inf_{\vv{v} \in
  W}(\dwidth{\vv{v}}{P_i}/\dwidth{\vv{v}}{P})$.  Since each $P_i$
satisfies unit-width constraint $W$, there exists
$\vv{u}_i \in \usetp$ such that $\dwidth{\vv{u}_i}{P_i} \le 1$.  Thus,
$1 \le \sum_{i=1}^m \inf_{\vv{v} \in
  W}(\dwidth{\vv{v}}{P_i}/\dwidth{\vv{v}}{P}) \le \sum_{i=1}^m
(\dwidth{\vv{u}_i}{P_i}/\dwidth{\vv{u}_i}{P}) \le \sum_{i=1}^m (1 /
\dwidth{\vv{u}_i}{P})$.

We analyze two cases depending on whether $\dwidth{\vv{v}}{P}$ attains
a minimum over $W$.  If it does, we have
$1 \le \sum_{i=1}^m (1 / \min_{\vv{v}\in W} \dwidth{\vv{v}}{P})$, and
$\min_{\vv{v}\in W} \dwidth{\vv{v}}{P} \le m$.  As $m$ is an integer,
$\min_{\vv{v}\in W} \lceil \dwidth{\vv{v}}{P} \rceil \le m$.  If no
minimum is attained over $W$,
$1 < \sum_{i=1}^m (1 / \inf_{\vv{v}\in W} \dwidth{\vv{v}}{P})$ and
$\inf_{\vv{v}\in W} \dwidth{\vv{v}}{P} < m$.  Then
$\min_{\vv{v}\in W} \lceil \dwidth{\vv{v}}{P} \rceil = \inf_{\vv{v}\in
  W} \lceil \dwidth{\vv{v}}{P} \rceil \le m$.  Both cases contradict
our assumption, and thus
$\nopt{P} = \min_{\vv{v}\in W} \lceil\dwidth{\vv{v}}{P}\rceil$.

Let $\vv{u} \in W$ be a vector minimizing
$\lceil \dwidth{\vv{u}}{P} \rceil$, computed in $O(n)$ time for
$W = \usetp$~\cite{Houle1985}.  Assume that $\vv{u}$ is given.  For
the vertices of $P$ given in counterclockwise order, such $m-1$
parallel cuts can be computed in $O(\min\{n, m\log\frac{n}{m}\})$
time~\cite{Chung2022}.  Since $x/(1+x) < \log (1+x) < x$ for all
$x>0$, $\min\{n, m\log\frac{n}{m} \} = \Theta(m\log (1+\frac{n}{m}))$.
For fixed $n$, $f(m)= m\log(1+\frac{n}{m})$ starts at $\Theta(\log n)$
when $m = 1$ and increases monotonically, approaching $\Theta(n)$ as
$m \to \infty$.
\begin{corollary}\label{cor:convex.opt.partition}
  Let $P$ be a convex polygon with $n$ vertices, and let
  $W, U\subseteq \usetp$ be sets of unit vectors such that
  $\overline{W} \subseteq U$.  Then an optimal partition for the
  problem $\prob{P,W,U}$ is achieved by equally spaced parallel cuts
  orthogonal to $\vv{u} \in W$ that minimizes
  $\lceil \dwidth{\vv{u}}{P} \rceil$. Given such a direction $\vv{u}$,
  the partition can be computed in
  $O(\dwidth{\vv{u}}{P} \log (1+\frac{n}{\dwidth{\vv{u}}{P}}))$ time.
\end{corollary}
\subparagraph{Remarks.}  We work in the Real RAM model, which
  supports unit-cost arithmetic ($+,-,\times,\div$) and comparisons on
  real numbers.
  Our algorithms perform integer rounding via comparisons.

\section{Concluding remarks}
We studied the minimum partition problem $\prob{P,W,U}$ 
for a simple polygon $P$ and direction sets $W, U\subseteq \usetp$. 
We provided necessary and sufficient conditions for the existence of feasible partitions, 
along with a decision algorithm for checking feasibility.

Our main contribution is an analysis of the monotonicity of the minimum partition number under polygon containment.
  To this end, we introduced the notion of $\oset$-convexity (for $\oset \subseteq \usetp$) of a subpolygon with respect to a container, 
  as a sufficient condition. 
  The monotonicity holds in two cases:
  (1) when $Q$ is $U$-convex with respect to $P$ in the guillotine case, and 
  (2) when $Q$ is $\overline{W}$-convex in the non-guillotine case. 
  The reconfiguration technique in the non-guillotine case also leads to a Bang-type theorem for partitions of convex bodies.

  We identify two open problems concerning the intractability 
  of the minimum partition problem under constraints $W$ and $U$.
  \begin{enumerate}
    \item $P$ is a hole-free simple polygon, with $U = \overline{W}$ and $|U| = |W| = 2$.  
    \item $P$ is a simple polygon with holes, with $U = \overline{W}$ and $|U| = |W| = 1$.  
  \end{enumerate}

  For hole-free polygons (Question 1), we may assume, without loss of generality, that $U$ and $W$ consist of orthogonal vectors, 
  by applying an appropriate linear transformation to $P$. 
  Abrahamsen and Stade~\cite{Abrahamsen2024} 
  proved NP-hardness of the AND version of the problem, 
  where each piece must have both horizontal and vertical width at most 1. 
  In contrast, our variant imposes the OR condition, where each piece must have either horizontal or vertical width at most 1. 
  Determining whether this OR version is NP-hard remains an open problem.

  For the case with holes (Question 2), 
  Buchin et al.~\cite{Buchin2021} and Worman~\cite{Worman2003} provide strong evidence 
  that our variant is NP-hard.  
  They studied 
  the minimum $\alpha$-fat partition problem where each piece has diameter at most $\alpha >0$, 
  without allowing Steiner points. 
  In our setting, arbitrary Steiner points are permitted, 
  which may increase the complexity of the problem,    
  whereas the restriction to parallel cuts may decrease it.  
  As in their works, 
  a common strategy is to reduce the problem to 
  planar 3-SAT in order to 
  show the intractability of partition problems for polygons with holes. 
  Thus, an attempt to reduce our variant to planar 3-SAT 
  appears to be a promising direction.

\end{document}